\documentclass[12pt]{article}
\usepackage{a4wide}
\usepackage{amssymb}
\usepackage{mathrsfs}
\begin{document}
{\renewcommand{\thefootnote}{\fnsymbol{footnote}}
\hfill  IGC--08/11--1\\
\medskip
\begin{center}
{\LARGE Gauge invariant cosmological perturbation equations with
corrections from loop quantum gravity}\\
\vspace{1.5em}
Martin Bojowald

\vspace{0.25em}

Institute for Gravitation and the Cosmos,
The Pennsylvania State University,\\
104 Davey Lab, University Park, PA 16802, USA\\

\vspace{0.7em}

Golam Mortuza Hossain

\vspace{0.25em}

Department of Mathematics and Statistics, University of New Brunswick,\\
Fredericton, NB E3B 5A3, Canada

\vspace{0.7em}

Mikhail Kagan

\vspace{0.25em}

Department of Science and Engineering, The Pennsylvania State
University, Abington\\
1600 Woodland Road, Abington, PA 19116, USA\\

and

Institute for Gravitation and the Cosmos,
The Pennsylvania State University,\\
104 Davey Lab, University Park, PA 16802, USA\\

\vspace{0.7em}

S.~Shankaranarayanan

\vspace{0.25em}

Institute of Cosmology and Gravitation, University
of Portsmouth, Mercantile House, Portsmouth~P01 2EG, U.K.

\vspace{1em}
\end{center}
}

\setcounter{footnote}{0}

\newcommand{\proofend}{\raisebox{1.3mm}{\fbox{\begin{minipage}[b][0cm][b]{0cm}
\end{minipage}}}}
\newenvironment{proof}{\noindent{\it Proof:} }{\mbox{}\hfill \proofend\\\mbox{}}
\newenvironment{ex}{\noindent{\it Example:} }{\medskip}
\newenvironment{rem}{\noindent{\it Remark:} }{\medskip}

\newtheorem{theo}{Theorem}
\newtheorem{lemma}{Lemma}
\newtheorem{defi}{Definition}

\newcommand{\case}[2]{{\textstyle \frac{#1}{#2}}}
\newcommand{\lP}{l_{\mathrm P}}
\newcommand{\be}{\begin{equation}}
\newcommand{\ee}{\end{equation}}
\newcommand{\bq}{\begin{eqnarray}}
\newcommand{\eq}{\end{eqnarray}}

\newcommand*{\R}{{\mathbb R}}
\newcommand*{\N}{{\mathbb N}}
\newcommand*{\Z}{{\mathbb Z}}
\newcommand*{\Q}{{\mathbb Q}}
\newcommand*{\C}{{\mathbb C}}

\newcommand{\md}{{\mathrm{d}}}
\newcommand{\tr}{\mathop{\mathrm{tr}}}
\newcommand{\sgn}{\mathop{\mathrm{sgn}}}
\newcommand{\pd}[2][]{\frac{\partial #1}{\partial #2}}
\def\f{\frac}
\def\t{\tilde}
\def\H{{\mathscr H}}
\def\h{{\cal H}}
\def\l{{\cal L}}
\newcommand{\D}{{\partial}}

\begin{abstract}
 A consistent implementation of quantum gravity is expected to change
 the familiar notions of space, time and the propagation of matter in
 drastic ways. This will have consequences on very small scales, but
 also gives rise to correction terms in evolution equations of modes
 relevant for observations. In particular, the evolution of
 inhomogeneities in the very early universe should be affected. In
 this paper consistent evolution equations for gauge-invariant
 perturbations in the presence of inverse triad corrections of loop
 quantum gravity are derived. Some immediate effects are pointed out,
 for instance concerning conservation of power on large scales and
 non-adiabaticity. It is also emphasized that several critical
 corrections can only be seen to arise in a fully consistent treatment
 where the gauge freedom of canonical gravity is not fixed before
 implementing quantum corrections. In particular, metric modes must be
 allowed to be inhomogeneous: it is not consistent to assume only
 matter inhomogeneities on a quantum corrected homogeneous background
 geometry.
 In this way, stringent consistency conditions arise for possible
 quantization ambiguities which will eventually be further constrained
 observationally.
\end{abstract}

\section{Introduction}

General relativity describes the structure and dynamics of space-time
by Einstein's equation for the space-time metric, which gives rise to
a wide range of phenomena in cosmology and astrophysics. In several
regimes, especially at high densities, one expects quantum gravity to
be crucial and to provide important correction terms to the classical
equations. By now, canonical quantum gravity in the loop approach has
progressed to the extent that perturbative calculations for the
behavior of inhomogeneities around Friedmann--Robertson--Walker
space-times can be performed at an effective level. This provides
cosmological applications, but by addressing the anomaly problem it
also sheds light on fundamental aspects such as quantum space-time
structure.

The classical equations of motion constitute an overdetermined set,
whose consistency is ensured by general covariance. When attempting to
include correction terms in these equations as they may be suggested
by quantum gravity, the consistency conditions must be respected. Only
an anomaly-free quantization, where consistency resulting from general
covariance is preserved, can lead to quantum equations of motion with
the correct number and behavior of degrees of freedom.

This issue becomes pressing already at the level of linearized
equations as they are used in cosmology of the early universe. Once
inhomogeneities are included as perturbations around an expanding
Friedmann--Robertson--Walker space-time, there are more equations than
unknowns which requires consistent forms of all terms in the
equations. Without inhomogeneities, consistency is automatically
satisfied: there is a single constraint which is always preserved by
the evolution it generates. Quantum corrections to homogeneous models
can therefore be implemented easily as it has been done often to
suggest diverse effects and scenarios \cite{LivRev}. However, just
putting quantum corrected solutions for a homogeneous background into
classical perturbation equations in general results in inconsistent
equations: the corrected background equations can no longer be
compatible with all the terms in the perturbation equations. Thus, a
direct treatment of inhomogeneities and the corrections they acquire
in a quantum theory of gravity is required.

Perturbations around homogeneous models then test their robustness,
demonstrate whether a particular form of quantum corrections can be
realized in a generally covariant way, and provide consistent sets of
equations of motion whose solutions can be analyzed for the
phenomenology and potentially observable effects they imply. From the
perspective of canonical quantum gravity, the consistency issue of
effective equations has been described in \cite{ConstraintAlgebra},
and it has been demonstrated that there is a correction expected from
loop quantum gravity \cite{Rov,ALRev,ThomasRev} which non-trivially
changes the classical equations in a consistent way at a perturbative
level.
(This correction results from inverse triad operators in
Hamiltonians \cite{QSDI,QSDV}. Other expected corrections from loop
quantum gravity, such as holonomy corrections, have not yet resulted
in consistent equations outside the reduced setting of homogeneous
models.)
In this paper, we derive the corresponding gauge invariant
perturbations and the equations of motion they satisfy.

Earlier work \cite{InhomEvolve,HamPerturb} had already led to quantum
corrections to Einstein's equation governing linear cosmological
perturbations. As a result, enhancement effects of quantum corrections
during long cosmic evolution times were suggested based on the
observation that super-horizon curvature perturbations were not
preserved, unlike classically, but had a growing mode. Other forms of
perturbation equations were used, for instance, in
\cite{ScaleInvLQC,CCPowerLoop,SuperInflLQC,Cosmography,HolonomyInfl}.
Such effects may lead to observable imprints in the cosmic background
radiation even though individual quantum-gravitational correction
terms at sub-Planckian densities are small.  The consideration,
however, was restricted to the scalar mode in the longitudinal gauge,
i.e.\ diagonal metric perturbations. While classically this procedure
proves to be equivalent to a non-gauge fixed derivation, it leads to
inconsistencies at the effective level in the presence of quantum
corrections. Specifically, the effective equations resulting from the
quantization of a gauge-fixed system in general are incompatible with
each other, and gauge-fixing eliminates the freedom required to see
systematically how the terms of consistent equations must be arranged.
A non-gauge fixed treatment is thus necessary to evaluate all
consistency conditions and to determine the gauge-invariant equation
of motion for the curvature perturbation. (Similarly, if one uses only
matter perturbations on a homogeneous gravitational background one is
implicitly using a gauge-fixed treatment. This is in general
inconsistent if no care is taken concerning the specific correction
terms and the meaning of matter perturbations in relation to
gauge-invariant quantities. Sometimes a ``separate universe'' picture
\cite{SeparateUniverse,SeparateUniverseII,SeparateUniverseIII} is
used, arguing that at least for large scale modes quantum corrections
to the homogeneous background equations should be sufficient to
determine the evolution of inhomogeneities. But also here, as we will
see in examples later, not all features will be visible based solely
on homogeneous models without a full consistency analysis once quantum
gravity corrections are included.)

A simple counting shows that the three independent
functions describing scalar (gravity and matter) perturbations are
subject to five equations of motion (see the Table \ref{TabEq}). In general,
such a system would be over-determined unless the equations are not
independent. In gravity, they can be split into two types: i)
evolution equations of second order in time derivatives and ii)
constraint equations of lower order. Constraints restrict the initial
data, and if they are preserved under evolution the system of
equations is consistent. This is guaranteed automatically when
equations of motion are obtained variationally from a covariant
action, and is therefore satisfied for the classical equations. If
quantum corrections are derived in a Hamiltonian approach, however,
consistency is ensured only if the quantization is anomaly-free. The
consistency of the resulting equations is thus tightly related to the
closure of the constraint algebra. While the algebra of the
constraints obtained from the classical Einstein-Hilbert action by a
Legendre transformation is closed or, in Dirac's terms
\cite{DirQuant}, the Hamiltonian and diffeomorphism constraints are of
first class, this property may not sustain quantization. A consistent
gauge invariant formulation of quantum-corrected equations of motion
is possible only if the quantization is anomaly free, i.e.\ if the
constraints remain first class.

As shown in \cite{vector,tensor}, standard loop quantization under
very mild assumptions leads to a non-anomalous constraint algebra for
vector and tensor modes. (At the linear level mode decomposition does
not interfere with quantization, and quantum corrections to scalar,
vector and tensor modes can be studied independently.)  In
\cite{ConstraintAlgebra} it was analyzed what types of (non-anomalous)
quantum corrections are allowed for the scalar mode, obtaining the
anomaly-freedom conditions ensuring a first class system. Once closure
of the constraint algebra is provided, the formulation of gauge
invariant equations of motion becomes possible, as developed in this
paper.

We start with reviewing the correspondence between the canonical and
covariant equations of motion, then derive the gauge invariant
variables and finally obtain the gauge invariant quantum corrected
linear Einstein equations. As we will see, consistency requires
certain features of the corrected perturbation equations and of the
gauge invariant variables, which could not be seen in gauge-fixed
formulations. Several immediate consequences are discussed in the
Sec.~\ref{s:Qualitative} and further in the conclusions, which also
exhibit the final quantum corrected perturbation equations. Readers
interested primarily in applications may turn directly to these
sections.

\section{Basic variables and equations}

When a classical theory is quantized, the choice of basic variables
often matters. While there are many equivalent formulations of
classical physics, all related to each other by canonical
transformations, such maps are rarely implementable as exact unitary
transformations when quantized. This gives rise to different
inequivalent quantizations of the same classical theory, and it can
even prevent one from constructing a quantization in a particular
classical formulation of the theory: there may be no Hilbert space
representation where a certain choice of basic classical phase space
variables will become well-defined operators.

In loop quantum gravity \cite{Rov,ALRev,ThomasRev}, the principle of
background independence, which requires that well-defined operators do
not refer to a metric other than the physical one to be turned into
operators, distinguishes a special class of basic variables. In field
theories such as general relativity, it is not field values at single
points which can become well-defined operators, but only ``smeared''
versions obtained after integrating over spatial regions. Such
integrations ensure that the operator-valued distributions, which
field values would correspond to, become well-defined operators which
can be multiplied to construct composite operators from them. Since
the physical fields of canonical general relativity are spatial tensor
fields, they cannot directly be integrated in a coordinate independent
manner. Moreover, integration measures are not provided automatically
because only the physical metric could be used on a curved manifold,
but this metric itself is being turned into an operator. If one uses
connection variables and densitized vector fields as canonical
objects, however, their transformation properties ensure that they can
be integrated over curves and surfaces, respectively, without
requiring any additional integration measure. The resulting holonomies
and fluxes then become well-defined operators in the quantum
representation underlying loop quantum gravity \cite{LoopRep}.

This representation has characteristic properties which are implied by
the choice of basic fields and their smearing. In particular,
operators for spatial geometry such as the fluxes themselves or areas
and volumes acquire discrete spectra \cite{AreaVol,Area,Vol2}. This,
in turn, determines how these basic operators can appear in composite
ones such as Hamiltonians \cite{RS:Ham,QSDI,QSDV}. For instance, flux
operators having discrete spectra containing the eigenvalue zero do
not possess densely defined inverse operators. Since inverse triad
components appear in Hamiltonians, the lack of a direct quantization
entails quantum corrections in any effective Hamiltonian (see e.g.\
\cite{QuantCorrPert}), which will then also change the corresponding
evolution as well as gauge properties. In this paper, we derive such
equations precisely for corrections resulting from inverse triad
components.

\subsection{Perturbed variables}

To do so, we perform the perturbation analysis of inhomogeneities in
the basic variables underlying loop quantum gravity, such that we will
be using primarily a densitized triad $E^a_i$ instead of the spatial
metric $q_{ab}$ (satisfying $E^a_iE^b_i=q^{ab}\det q$). Moreover, in
this canonical setting the remaining components $N$ and $N^a$ of the
space-time metric
\begin{equation}\label{CanonMetric}
 \md s^2= -N^2\md t^2+ q_{ab}(\md x^a+N^a\md t) (\md x^b+N^b\md t)
\end{equation}
will not be dynamical but play the role of Lagrange multipliers of
constraints. In fact, their time derivatives do not appear in the
Einstein--Hilbert action, which can be written in the canonical form
\begin{equation}\label{action}
 S_{\rm EH}=  \int \md t\left[ \left(  \frac{1}{8\pi G}\int\md^3x
\dot{K}_a^iE^a_i \right) -G_{\rm grav}[\Lambda^i] -D_{\rm grav}[N^a]
  -H_{\rm grav}[N]
 \right]
\end{equation}
where $K_a^i$ is conjugate to $E^a_i$,
\begin{equation}
 \{K_a^i(x),E^b_j(y)\}= 8\pi G\delta_a^b\delta^i_j\delta^3(x,y)\,,
\end{equation}
and related to extrinsic curvature $K_{ab}$ by
$K_a^i=K_{ab}E^b_i/\sqrt{|\det E|}$.
The remaining terms are the diffeomorphism constraint
\begin{equation} \label{DiffeoConstraint}
D_{\rm grav}[N^a] = \frac{1}{8\pi G}\int_{\Sigma}\mathrm{d}^3xN^a
\left((\partial_a K_b^j - \partial_b K_a^j)E^b_j - K_a^j
\partial_b E^b_j \right)
\end{equation}
and the Hamiltonian constraint
\begin{equation} \label{HamConstClass}
H_{\rm grav}[N] = \frac{1}{16\pi G} \int_{\Sigma}\mathrm{d}^3x N
\frac{E_i^a E^b_j}{\sqrt{|\det
E|}}\left(F_{ab}^k{\epsilon^{ij}}_{k}-2(1+\gamma^{2})
K_a^{[i}K_b^{j]}\right)\,.
\end{equation}
Here,
\[
F_{ab}^k=2\D_{[a}(\Gamma+\gamma K)_{b]}^k+{\epsilon_{ij}}^{k}
(\Gamma+\gamma K)_a^i (\Gamma+\gamma K)_b^j
\]
is the curvature of the Ashtekar--Barbero connection
\cite{AshVar,AshVarReell} $A_a^i=\Gamma_a^i+\gamma K_a^i$ defined in
terms of the spin connection
\begin{equation}
 \Gamma_a^i=-\frac{1}{2}\epsilon^{ijk} E_j^b \left(\partial_a
E_b^k-\partial_b E_a^k+E_k^c E_a^l \partial _c E_b^l - E_a^k
\frac{\partial_b (\det E)}{\det E}\right)
\end{equation}
and $\gamma$ is the Barbero--Immirzi parameter
\cite{AshVarReell,Immirzi}.
These variables also appear in the Gauss constraint
\[
 G_{\rm grav}[\Lambda^i]= \int\md^3x \Lambda^i(\partial_a E^a_i+
\epsilon_{ijk} \Gamma_a^jE^a_k+ \gamma \epsilon_{ijk} K_a^jE^a_k)
\]
in (\ref{action}). This constraint will be solved explicitly by our
parameterization of variables at the linear level, and its gauge will
be fixed by a background triad. We can thus ignore it from now on.

For a perturbed metric of the form
\begin{equation}\label{MetricPert} \md s^2 =
a^2(\eta)\left(-(1+2\phi)\md \eta^2 +2\D_a B\md \eta \md x^a +
((1-2\psi)\delta_{ab}+2\D_a\D_b E)\md x^a\md x^b\right)\,,
\end{equation}
as it describes general scalar perturbations $(\phi,\psi,E,B)$ around
spatially flat Friedmann--Robertson--Walker models in a general gauge,
the background and perturbed triad in
\begin{equation} \label{TriadPert}
 E^a_i= \bar{E}^a_i+\delta E^a_i
\end{equation}
are given by
\begin{equation}\label{Triad}
\bar E_i^a=\bar p \delta_i^a \equiv a^2 \delta_i^a\quad,
\quad \delta E_i^a = -2\bar p \psi \delta_i^a +\bar p(\delta_i^a
\Delta-\D^a\D_i)E,
\end{equation}
where $\Delta$ is the Laplace operator on flat space, indices
run from 1 to 3, $a$ is the (background) scale factor and $\psi$
and $E$ describe the spatial part of the perturbed metric.
(We use standard notations where $E$ is one of the scalar modes
distinct from the full densitized triad $E^a_i$. The latter will
always be written with indices such that no confusion should arise.)
The other
two scalar metric perturbations $\phi$ and $B$ are related to the
perturbed lapse function and shift vector respectively
\begin{equation}\label{PertLapseShift}
\delta N=\bar N \phi \quad,
\quad N^a=\D^a B
\end{equation}
and will enter the extrinsic curvature components, also perturbed as
\begin{equation}\label{KPert}
 K^i_a=\bar{K}_a^i+\delta K_a^i =\bar{k}\delta_a^i+\delta K_a^i\,.
\end{equation}
This splitting, with the condition that $\delta E^a_i$ and $\delta
K_a^i$ do not have homogeneous modes in order to avoid double
counting, results in Poisson brackets
\begin{equation}\label{BasicPBGrav}
\left\{\bar k,\bar p\right\}=\f{8 \pi G}{3V_0}\quad, \quad
\left\{ \delta K_a^i(x) , \delta E_j^b(y) \right\}=
8\pi G \delta^i_j \delta_a^b\delta^3(x-y) \,.
\end{equation}

The homogeneous mode is defined by
\begin{equation}
 \bar{p}= \frac{1}{3V_0} \int E^a_i\delta^i_a\md^3x\quad,\quad
\bar{k}= \frac{1}{3V_0} \int K_a^i\delta_i^a\md^3x
\end{equation}
where we integrate over a bounded region of coordinate size
$V_0=\int\md^3x$ which could be over the whole space if it is compact,
or a sufficiently large region encompassing all the scales of
perturbations of interest. Although $V_0$, which depends on
coordinates as well as the choice we make for the integration region,
enters the definition of variables and their Poisson structure,
particular correction terms for observables will not depend on its
value.  Using the homogeneous modes, perturbations are defined by
(\ref{TriadPert}) and (\ref{KPert}).

The specific form of $\bar{K}_a^i$ and $\delta K_a^i$ in relation to
time derivatives of $\bar{p}$ and $\delta E^a_i$, analogously to the
triad fields in (\ref{Triad}), follows from the equations of
motion. Before deriving these relations we thus introduce the quantum
corrections we consider because they have a bearing on the form of
components of $K_a^i$.

\subsection{Quantum corrections}

The Hamiltonian constraint (\ref{HamConstClass}) contains a factor of
an inverse determinant of the densitized triad. This inverse cannot be
quantized directly because the integrated determinant itself is
quantized to an operator with zero in the discrete spectrum,
precluding the existence of an inverse operator. Nevertheless,
well-defined operators quantizing (\ref{HamConstClass}), including the
inverse triad, exist \cite{QSDI,InvScale}. However, the behavior of
expectation values of the operators differs on small length scales
from the classical behavior even in semiclassical states, which
implies the presence of a correction function $\alpha$ multiplying the
Hamiltonian density. This function must be scalar (of density weight
zero) to ensure the proper behavior of the integral, and it can depend
functionally on all the phase space variables in possibly non-local
ways. Such a general dependence would make an analysis of the
constraint algebra and of equations of motion intractable, and so we
have organized the calculations in \cite{ConstraintAlgebra} by first
assuming a primary correction function $\alpha(E^a_i)$ which depends
only on the triad and does so only in algebraic form. By itself, this
does not produce anomaly-free quantizations, which however do exist if
additional counter-terms are added containing new correction functions
whose relation to the primary correction is fixed by
anomaly-cancellation. These extra terms can be interpreted as arising
from a more complicated dependence of $\alpha$ on all the phase space
variables, which is derived systematically by this process.

Specifically, the quantum corrected Hamiltonian constraint derived in
\cite{ConstraintAlgebra} can conveniently be written as
\begin{equation}\label{H2Q}
H^Q = H^Q_{\rm grav} [\bar N] + H^Q_{\rm grav} [\delta N]
+H^Q_{\rm matt} [\bar N]+H^Q_{\rm matt} [\delta N],
\end{equation}
where the gravitational part is expanded by powers of inhomogeneities
$\delta E^a_i$, $\delta K_a^i$ and $\delta N$ as
\begin{eqnarray} \label{HamGravQ}
H_{\rm grav}^Q[\bar{N}] &:=& \frac{1}{16\pi G} \int\mathrm{d}^3x
\bar{N}\left[ \bar{\alpha}{\mathcal H}^{Q(0)} +
\alpha^{(2)}{\mathcal H}^{Q(0)} +
\bar{\alpha}{\mathcal H}^{Q(2)}\right], \nonumber\\
H_{\rm grav}^Q[\delta N] &:=& \frac{1}{16\pi G}\int\mathrm{d}^3x
 \delta N \left[\bar{\alpha}{\mathcal H}^{Q(1)}\right] ~,
\end{eqnarray}
and the matter Hamiltonian reads
\begin{eqnarray}\label{HamMatterQ}
 H^Q_{\rm matter}[\bar N]&=& \int_{\Sigma} \mathrm{d}^3x
\bar{N}\left[\left(\bar\nu\h_\pi^{Q(0)}+
\h_\varphi^{Q(0)}\right)+\left(\nu^{(2)}\h_\pi^{Q(0)}+
\bar\nu\h_\pi^{Q(2)}+\bar\sigma\h_\nabla^{Q(2)}+\h_\varphi^{Q(2)}\right)\right]~\nonumber\\
H_{\rm matter}^Q[\delta{N}] &=& \int\mathrm{d}^3x \delta N \left[
\bar{\nu}{\mathcal H}_\pi^{Q(1)} + {\mathcal H}_\varphi^{Q(1)}
\right] ~.
\end{eqnarray}
In all expansions, a bar is used to denote background quantities while
superscripts indicate the inhomogeneous order.  The Hamiltonian
densities in the expansions are given by
\cite{ConstraintAlgebra}
\begin{eqnarray} \label{HamGravQDens}
{\mathcal H}^{Q(0)} &=& -6\bar{k}^2\sqrt{\bar p}~,\nonumber\\
{\mathcal H}^{Q(1)} &=& -4(1+f) \bar{k}\sqrt{\bar{p}}
\delta^c_j\delta K_c^j -(1+g)\frac{\bar{k}^2}{\sqrt{\bar{p}}}
\delta_c^j\delta E^c_j +\frac{2}{\sqrt{\bar{p}}}
\partial_c\partial^j\delta E^c_j  ~,
\nonumber\\  {\mathcal H}^{Q(2)} &=& \sqrt{\bar{p}} \delta
K_c^j\delta K_d^k\delta^c_k\delta^d_j - \sqrt{\bar{p}} (\delta
K_c^j\delta^c_j)^2 -\frac{2\bar{k}}{\sqrt{\bar{p}}} \delta
E^c_j\delta K_c^j
\\
&& \quad -\frac{\bar{k}^2}{2\bar{p}^{3/2}} \delta E^c_j\delta
E^d_k\delta_c^k\delta_d^j +\frac{\bar{k}^2}{4\bar{p}^{3/2}}(\delta
E^c_j\delta_c^j)^2 -(1+h)\frac{\delta^{jk}
}{2\bar{p}^{3/2}}(\partial_c\delta E^c_j) (\partial_d\delta E^d_k)
~.\nonumber
\end{eqnarray}
for gravity and by
\begin{eqnarray}\label{HamMatterQDens}
{\mathcal H}_\pi^{Q(0)} &=&
\frac{\bar{\pi}_{\bar\varphi}^2}{2\bar{p}^{3/2}} \quad,\quad
{\mathcal H}_\nabla^{Q(0)} = 0 \quad,\quad {\mathcal
H}_\varphi^{Q(0)} = \bar{p}^{3/2} V(\bar{\varphi}) ~,
\nonumber\\{\mathcal H}_\pi^{Q(1)} &=& (1+f_1)\frac{\bar{\pi}
\delta{\pi}}{\bar{p}^{3/2}}
-(1+f_2)\frac{\bar{\pi}^2}{2\bar{p}^{3/2}} \frac{\delta_c^j
\delta E^c_j}{2\bar{p}}\nonumber\\
{\mathcal H}_\nabla^{Q(1)} &=& 0\\
 {\mathcal H}_\varphi^{Q(1)} &=&
\bar{p}^{3/2}\left( (1+f_3)V_{,\varphi}(\bar{\varphi})
\delta\varphi +V(\bar{\varphi}) \frac{\delta_c^j \delta
E^c_j}{2\bar{p}}\right)\nonumber\\
 \h^{Q(2)}_{\pi}&=&
(1+g_1)\frac{{\delta{\pi}}^2}{2\bar{p}^{3/2}}
-(1+g_2)\frac{\bar{\pi} \delta{\pi}}{\bar{p}^{3/2}}
\frac{\delta_c^j \delta E^c_j}{2\bar{p}}
\nonumber+\frac{1}{2}\frac{\bar{\pi}^2}{\bar{p}^{3/2}} \left(
(1+g_3)\frac{(\delta_c^j \delta E^c_j)^2}{8\bar{p}^2}
+\frac{\delta_c^k\delta_d^j\delta E^c_j\delta E^d_k}{4\bar{p}^2}
\right) \nonumber\\
\h^{Q(2)}_{\nabla}&=&\frac{1}{2}(1+g_5)\sqrt{\bar{p}}\delta^{ab}\partial_a\delta
\varphi
\partial_b\delta \varphi\nonumber\\
\h^{Q(2)}_{\varphi} &=&\bar{p}^{3/2} \left[(1+g_6)\frac{1}{2}
V_{,\varphi\varphi}(\bar{\varphi}) {\delta\varphi}^2 +
V_{,\varphi}(\bar{\varphi}) \delta\varphi \frac{\delta_c^j \delta
E^c_j}{2\bar{p}}+ V(\bar{\varphi})\left( \frac{(\delta_c^j \delta
E^c_j)^2}{8\bar{p}^2} -\frac{\delta_c^k\delta_d^j\delta
E^c_j\delta E^d_k}{4\bar{p}^2} \right)\right] \,. \nonumber
\end{eqnarray}
for matter.

The $\bar{p}$-dependent functions $\bar{\alpha}$, $\bar{\nu}$ and
$\bar{\sigma}$ are primary correction functions whose origin is the
presence of inverse triad operators in a constraint operator. Their
form can be computed in isotropic models \cite{InvScale,Ambig,ICGC} or
with certain gauge assumptions for inhomogeneous states
\cite{BoundFull,QuantCorrPert}. Classically, we have
$\bar{\alpha}=\bar{\nu}=\bar{\sigma}=1$, while there can be strong
deviations from this value for small values of elementary flux
variables which quantize the densitized triad. This deep quantum
regime is difficult to control, however, and the derivations in
\cite{ConstraintAlgebra} of an anomaly-free constraint algebra are
valid for primary correction functions of the form
\begin{equation}\label{AlphaHom}
 \bar{\alpha}(a)=
 1+c_{\alpha}\left(\frac{\ell_{\rm P}^2}{a^2}\right)^{n_{\alpha}}+\cdots
\end{equation}
which are perturbative in the Planck length $\ell_{\rm
  P}=\sqrt{G\hbar}$, i.e.\ $n_{\alpha}>0$. Explicit values for
  coefficients $c_{\alpha}$, $c_{\nu}$ and $c_{\sigma}$, which are
  generically positive such that $\bar{\alpha}(a)>1$ in perturbative
  regimes, as well as the exponents $n_{\alpha}$, $n_{\nu}$ and
  $n_{\sigma}$ can be derived from specific quantizations, but they
  are subject to quantization ambiguities.\footnote{Even the isotropic
  quantization used for the background evolution is subject to
  quantization ambiguities. Uniqueness results of the quantum dynamics
  can be obtained only based on ad-hoc assumptions, and they sometimes
  occur as a result of incorrect implementations of quantization
  schemes. The source of ambiguities is the representation of
  operators, such as inverse triad operators used here, but also the
  underlying refinement behavior of a discrete state underlying the
  quantum evolution \cite{InhomLattice,CosConst}; see also
  \cite{CHRev} and the appendix of \cite{ConstraintAlgebra}.  Both
  ingredients combine to determine the values of $c_{\alpha}$ and
  $n_{\alpha}$. Note also that $c_{\alpha}$, when correctly derived
  using lattice refinement, is coordinate dependent in such a way that
  the combination with the scale factor in (\ref{AlphaHom}) is scaling
  independent.}
One purpose of deriving
anomaly-free versions of the constraints is to provide consistency
conditions among some of these values, fixing some quantization
ambiguities.

For an anomaly-free quantization in a
gauge-independent manner, the presence of these primary correction
functions requires counter-terms with coefficients $f$, $g$ and $h$ as
well as $f_i$ and $g_i$ which also depend on $\bar{p}$ in a way fixed
by anomaly cancellation conditions. For the situation under consideration
where the matter sector consists of a scalar field with a non-trivial
potential, we have
\begin{eqnarray}\label{fgh}
 2f' \bar{p} &=& -\frac{\bar{\alpha}'\bar{p}}{\bar{\alpha}}\\
g&=& -2f \label{g}\\
f_1 &=& f -\frac{\bar{\nu}'\bar{p}}{3\bar{\nu}}
 \label{f1}
\end{eqnarray}
and
\begin{equation} \label{alpha2}
\frac{\partial\alpha^{(2)}}{\partial(\delta E^a_i)} (\delta^c_j
\delta^a_i-\delta^a_j \delta^c_i) = \frac{\alpha^\prime}{3p}\delta
E^c_j, \quad \frac{\partial\nu^{(2)}}{\partial(\delta E^a_i)}
(\delta^c_j \delta^a_i-\delta^a_j \delta^c_i) =
\frac{\nu^\prime}{3p}\delta E^c_j\,.
\end{equation}
Here and in what follows, primes denote derivatives by $\bar{p}$.
Moreover, we have
\begin{equation}
 \bar{\alpha}^2=\bar{\nu}\bar{\sigma}\,.
\end{equation}
Other consistency conditions will be recalled later from
\cite{ConstraintAlgebra} (also discussed in Appendix \ref{App:Anomaly}) whenever
they are being used. Classically,
all counter-terms vanish, e.g.\ $f=f_1=g=0$. With the consistency
conditions the system of corrected constraints is anomaly-free to the
perturbative orders considered, which is linear in inhomogeneities
(requiring second order expansions of the constraints which generate
linear equations of motion) as well as leading order in the
corrections of (\ref{AlphaHom}). The latter assumption of
perturbativity implies that we ignore terms such as
$(\bar{\alpha}-1)^2$, $(p\md\bar{\alpha}/\md p)^2$ or $f^2$ compared to
$\bar{\alpha}-1$.

With the corrected Hamiltonian, we can derive the equations of motion
it generates. From $\dot{\bar{p}}=\{\bar{p},H^Q_{\rm grav}[N]\}$, for
instance, we obtain the background part
$\bar{K}_a^i=\bar{k}\delta_a^i$ of extrinsic curvature where $\bar{k}$
is related to the conformal Hubble parameter by
\be\label{HubbleK} \bar \alpha \bar k = \H
\equiv \f{\dot{\bar p}}{2\bar{p}}\,.
\ee

The choice of the background lapse function $\bar{N}=a$, used to
derive (\ref{HubbleK}), corresponds to the conformal time $\eta$ whose
derivative we denote by a dot.  In general, the total time derivative
of an arbitrary phase space function is given by its Poisson bracket
with $H^Q[N]+D[N^a]$ parameterized by the total lapse $N\equiv\bar N
+\delta N$ and shift $N^a\equiv\bar N^a +\delta N^a$. Nonetheless, for
a background quantity the Poisson bracket above, using only
$H^Q[\bar{N}]$, coincides with the conformal (background) time
derivative up to the second perturbative order.

Similarly, the form of the perturbation $\delta K_a^i$ can be deduced
from Hamilton's equation for $\delta \dot E_i^a$. Namely, using
\begin{equation} \label{deltaEdot}
\delta \dot E_i^a \equiv \{\delta E_i^a, H^Q[N]+D[N^a]\}
\end{equation}
along with (\ref{Triad}) and (\ref{HubbleK}), we obtain
\begin{equation}\label{deltaK}
\bar{\alpha}\delta K^i_a=-\delta^i_a\left[\dot \psi+\H(\psi
+\phi(1+f))\right]+\D_a\D^i\left[\H E-(B-\dot E)\right],
\end{equation}
where the counter-term $f(\bar p)$ appears.

Matter is represented by a scalar field $\varphi=\bar
\varphi+\delta \varphi$ with potential
$V(\varphi)$ and its conjugate momentum $\pi=\bar\pi+\delta\pi$.
As before, we use the equations of motion
\begin{equation} \label{phibardot}
\dot{\bar \varphi} \equiv \{\bar\varphi, H^Q[N]+D[N^a]\},\quad \delta \dot
\varphi \equiv \{\delta \varphi, H^Q[N]+D[N^a]\}
\end{equation}
to express the field momentum as
\begin{equation}\label{deltapi}
\bar\pi=\dot{\bar\varphi}\f{\bar p}{\bar\nu},\quad
\delta\pi=\f{\bar{p}}{\bar\nu}\left(\left(\delta\dot\varphi-
\dot{\bar\varphi}(1+f_1)\phi\right)(1-g_1)+\dot{\bar\varphi}\f{\delta
E_i^a\delta_a^i}{2\bar{p}}\right)\,.
\end{equation}

Before proceeding to gauge transformations in the next section we note
the relation between the canonical and covariant equations of motion,
summarized in Table \ref{TabEq}.
\begin{table}
\begin{center}
\begin{tabular}{|c|c|}
\hline
{\mbox{{\bf Covariant Equations}}} & \mbox{{\bf Canonical Equations}}\\
\hline%
 Background Friedmann&  Background Hamiltonian Constraint \\\hline%
Background Raychaudhuri& $\dot{\bar k}$ \& $\dot{\bar p}$  \\\hline%
Background Klein-Gordon& $\dot{\bar \varphi}$ \& $\dot{\bar \pi}$ \\\hline%
Perturbed Einstein $_T^T$& Perturbed Hamiltonian Constraint   \\\hline%
Perturbed Einstein $_T^S$&  Perturbed Diffeomorphism Constraint  \\\hline%
Perturbed Einstein $_S^S$&   $\delta \dot K$ \& $\delta \dot E$ \\ \hline%
Perturbed Klein-Gordon&   $\delta \dot\varphi$ \& $\delta \dot\pi$ \\%
\hline
\end{tabular}
\caption{Table of correspondence between the background and perturbed
canonical equations and the covariant equations of general
relativity. In the covariant column, the subscripts `S' and `T' stand
for spatial and temporal components respectively. In the canonical
framework, equations are of two types: constraint equations and
dynamical (Hamilton's) equations for the time derivatives of canonical
pairs. Note that both `SS'- and `$\delta \dot K$ \& $\delta \dot
E$'-equations are tensorial. Indices of $\delta K_a^i$ and $\delta
E_i^a$ have been suppressed for simplicity.\label{TabEq}}
\end{center}
\end{table}
There are three background equations, only
two of which are independent, for the two unknown functions: scale
factor $a(\eta)$ and matter scalar field $\bar \varphi(\eta)$
depending on the conformal time $\eta$. Those are the Friedmann,
Raychaudhuri and Klein-Gordon equations  \cite{QuantCorrPert}
\begin{eqnarray}
\H^2&=&\f{8\pi G}{3}\bar{\alpha}\left(\f{\dot{\bar
\varphi}^2}{2\bar\nu}+\bar{p}V(\bar\varphi)\right)\label{FriBG},\\
\dot\H&=&\H^2\left(1+\f{\bar\alpha^\prime \bar p}{\bar\alpha}\right)-4\pi
G\f{\bar\alpha}{\bar\nu}\dot{\bar\varphi}^2\left(1-\f{\bar\nu^\prime
\bar p}{3\bar\nu}\right)\label{RayBG},\\
\ddot{\bar\varphi}&+&2\H\dot{\bar\varphi}\left(1-\f{\bar\nu^\prime {\bar p}}{\bar\nu}
\right)+\bar\nu
\bar p V_{,\varphi}(\bar\varphi)=0\label{KGBG},
\end{eqnarray}
where the prime indicates a derivative with respect to $\bar p$.
These equations are listed in the left column and correspond to the
background Hamiltonian constraint and two pairs of dynamical
(Hamilton's) equations. Each pair of the first order equations is to
be combined into a single second order equation.

More generally, covariant equations which are less than second order
(with respect to the conformal time derivative) correspond to
constraint equations in the canonical formalism. They should be viewed
as restrictions on the initial conditions. As mentioned earlier, for
consistent and unambiguous (gauge-invariant) evolution such
constraints must be preserved by the dynamical equations. In the
canonical language, this property of constraints is translated into
the requirement of closure of the constraint algebra, as it is
analyzed for the corrected constraints in \cite{ConstraintAlgebra}. As
a result, relations such as (\ref{fgh}), (\ref{f1}), (\ref{g}) and
(\ref{alpha2}) between the correction functions have to be satisfied
as conditions for higher order terms of primary correction functions.

\section{Gauge transformations}

In classical relativity, it is the Lie derivative which provides the
form of gauge transformations of the fields such as metric components,
corresponding to changes of coordinates. After quantization it is no
longer clear what the analog of these gauge transformations should be,
given that the underlying space-time notion would have to be
determined from the quantum theory itself. In most approaches to
quantum gravity, one does not expect the fundamental space-time
picture to be described by a smooth manifold. Here, an advantage of
the canonical formulation is that gauge transformations are directly
generated as Poisson brackets of the fields with the
constraints. Classically, this reproduces the formulas obtained by Lie
derivatives, and it can directly be extended to canonical quantum
gravity capturing changes to the quantum space-time structure. With
corrections (\ref{H2Q}) to the classical constraints
(\ref{HamConstClass}), it is not only equations of motion but also the
form of gauge transformations which changes. Thus, gauge invariant
combinations of the perturbations take different forms than
they do classically.

However, as we saw not all the space-time metric components are
dynamical phase space variables, and only the gauge transformations
for the spatial metric, or $E^a_i$ and $K_a^i$, will be determined
straightforwardly. In this section we derive these transformations and
show how also the transformations of the remaining components $N$ and
$N^a$, or $\phi$ and $B$ in the scalar perturbations related to the
Lagrange multipliers, can be obtained.

\subsection{Classical gauge-invariant variables}

It is instructive to introduce the canonical derivations and required
notions in the classical case first, after which we will directly
extend the expressions to those including quantum
corrections. (Canonical treatments of classical perturbations have
also been discussed in \cite{PertObsI,PertObsII,BKdustII}.)

In the covariant formulation, gauge transformations constitute
local infinitesimal coordinate transformations
\[
x^\mu\rightarrow\tilde x^\mu = x^\mu+\xi^\mu(x),
\]
generated by vector fields $\xi^\mu$. In a perturbative setting, the
infinitesimal field $\xi^{\mu}$ will be treated as a first order
perturbation. Under this coordinate transformation any tensor field
receives a correction equal to its Lie derivative along $\xi^\mu$.
The part of the transformation relevant for the scalar mode can be
parameterized by two scalar functions $\xi^0$ and $\xi$ such that
\[
\xi^\mu=(\xi^0,\D^a\xi)
\]
where $a$ indicates a spatial direction 1, 2 or 3.

There are four (spatial) scalar perturbations in a space-time
metric, $\phi$, $\psi$, $B$ and $E$ as they appear in the line
element (\ref{MetricPert})
in conformal time $\eta$. These perturbations are subject to two
independent gauge transformations by $\xi^0$ and $\xi$. We now briefly
recall how these transformations follow from changes of coordinates to
verify later that the classical canonical transformations produce the
correct form.

If only $\xi^0$ is non-zero, the coordinate transformation changes
$\eta$ to $\eta+\xi^0$, which for $\md\eta^2$ implies, to first order
in $\xi^0$,
\[
 \md(\eta+\xi^0)^2=\md\eta^2+2\dot{\xi}^0\md\eta^2+ 2\xi^0,_a \md\eta\md x^a
\]
and $a(\eta)^2$ changes to $a(\eta)^2(1+2\dot{a}\xi^0/a)$. Inserting
this in (\ref{MetricPert}), we read off the transformation formulas
\begin{equation}
 \phi\mapsto
\phi+\dot{\xi}^0+\H\xi^0\quad,\quad \psi\mapsto
\psi-\H\xi^0\quad,\quad B\mapsto B-\xi^0
\quad,\quad E\mapsto E
\end{equation}
where $\H=\dot{a}/a$ as in (\ref{HubbleK}).

If only $\xi$ is non-zero, $\md x_a\md x^a$ changes to
\[
 \md(x_a+\xi,_a)\md(x^a+\xi,^a)=
 \md x_a\md x^a+ 2\dot{\xi},_a\md\eta\md x^a+ 2\xi,_{ab}\md x^a\md x^b
\]
which yields
\begin{equation}
 \phi\mapsto \phi\quad,\quad \psi\mapsto
\psi\quad,\quad B\mapsto B+\dot{\xi}  \quad,\quad E\mapsto
E+\xi\,.
\end{equation}

We thus see that $B-\dot{E}$ is invariant under $\xi$-transformations
(or spatial diffeomorphisms) and changes to $B-\dot{E}-\xi^0$ under
$\xi^0$-transformations. Thus, the Bardeen variables \cite{Bardeen}
\begin{equation} \label{Bardeen}
 \Phi=\phi+ \H(B-\dot{E})+(B-\dot{E})^\cdot \quad\mbox{ and }\quad
\Psi=\psi-\H(B-\dot{E})
\end{equation}
are gauge invariant.
For a scalar field $\varphi$, the only change is under $\xi^0$ and
given by $\delta\varphi+\dot{\bar{\varphi}}\xi^0$. Here,
\begin{equation}
 \delta\varphi^{\rm GI}=\delta\varphi+\dot{\bar{\varphi}}(B-\dot{E})
\end{equation}
is gauge invariant.

In the canonical formulation, gauge transformations are generated by
the Hamiltonian and diffeomorphism constraints. The corresponding
lapse function and shift vector to be inserted are also first order
perturbations related to the infinitesimal vector field $\xi^\mu$ via
\be\label{LapseShift_Def_Cl} \delta N = \bar N\xi_0, \quad \delta
N^a = \D^a \xi,
\ee
which follow from the metric decomposition (\ref{CanonMetric})
in terms of the spatial metric $q_{ab}$, lapse function $N$ and shift
vector $N^a$.  With first order smearing functions, the
gauge-generating constraints are at least of second perturbative
order. From now on we will denote the gauge transformations of a
phase space quantity $X$ as
\be\label{Gauge_Trans}
\delta_{[\xi_0,\xi]}X\equiv \{X,H^{(2)}[\bar
N\xi_0]\}+\{X,D^{(2)}[\D^a\xi]\},
\ee
where
\be\label{H2gauge} H^{(2)}[\delta N]=\frac{1}{16\pi G}
\int\mathrm{d}^3x\, \delta N[\h^{(1)}+16\pi
G(\h_\pi^{(1)}+\h_\varphi^{(1)})] \ee
with the Hamiltonian densities given by
\begin{equation} \label{PertHamConst_Grav}
{\mathcal H}^{(1)} = -4 \bar{k}\sqrt{\bar{p}} \delta^c_j\delta
K_c^j -\frac{\bar{k}^2}{\sqrt{\bar{p}}} \delta_c^j\delta E^c_j
+\frac{2}{\sqrt{\bar{p}}}
\partial_c\partial^j\delta E^c_j
\end{equation}
and
\begin{equation} \label{SFHamConstH1}
{\mathcal H}_\pi^{(1)} = \frac{\bar{\pi}
\delta{\pi}}{\bar{p}^{3/2}} -\frac{\bar{\pi}^2}{2\bar{p}^{3/2}}
\frac{\delta_c^j \delta E^c_j}{2\bar{p}} , \quad {\mathcal
H}_\nabla^{(1)} = 0 ,\quad {\mathcal H}_\varphi^{(1)} =
\bar{p}^{3/2}\left( V_{,\varphi}(\bar{\varphi}) \delta\varphi
+V(\bar{\varphi}) \frac{\delta_c^j \delta E^c_j}{2\bar{p}}\right)\,.
\end{equation}
For the diffeomorphism constraint (\ref{DiffeoConstraint}), we have
the second order term
\be\label{D2gauge} D^{(2)}[\delta N^a]=\frac{1}{8\pi
G}\int_{\Sigma}\mathrm{d}^3x\,\delta N^a
\left[\bar{p}\partial_a(\delta^d_k \delta K^k_d)
-\bar{p}(\partial_k\delta K^k_a)- \bar{k} \delta_a^k(
\partial_d \delta E^d_k)+8\pi G(\bar \pi \D_a \delta\varphi)\right].
\ee
In what follows, we perform canonical gauge transformations on the
basic phase variables and demonstrate that this, too, results in the gauge
invariant combinations (\ref{Bardeen}).

        \subsubsection{Gauge transformations of basic variables}

We start by computing gauge transformations of the basic
phase space variables: $K_a^i$, $E_i^a$, $\varphi$ and $\pi$.
Using (\ref{Gauge_Trans}) with the constraints (\ref{H2gauge}) and
(\ref{D2gauge}) we obtain \bq\label{Gauge_Trans_Explicit}
\delta_{[\xi_0,\xi]}\delta K_a^i&=&\D^i\D_a(\xi_0+\bar k \xi)-\f{\bar k^2}{2}\xi_0 \delta_a^i + 4\pi G \left[\bar p V(\bar\varphi)-\f{\dot{\bar\varphi}^2}{2}\right]\xi_0\delta_a^i,\nonumber\\
\delta_{[\xi_0,\xi]}\delta E^a_i&=&2\bar k \bar p \xi_0 \delta_i^a + \bar p
(\delta_i^a \Delta\xi-\D^a\D_i\xi),\nonumber\\
\delta_{[\xi_0,\xi]}\delta\varphi&=&\f{\bar \pi}{\bar p}\xi_0,\nonumber\\
\delta_{[\xi_0,\xi]}\delta\pi&=&\bar \pi \Delta\xi-\bar p^2
V^\prime(\bar\varphi)\xi_0
\eq for the basic gravity and matter perturbations. Note that the
background lapse function has again been set to $\bar N = \sqrt{\bar p}$ for
gauge transformations of a metric in conformal time. It is also easy
to see that when acting upon the background quantities $\bar k$, $\bar
p$, $\bar \varphi$, and $\bar \pi$ these gauge transformations will
generate only second order contributions. Hence in the equations of
motion the background phase space variables can be treated as
gauge-invariant up to the desired order.

The perturbative treatment introduces a subtlety in the interpretation
of transformations: In the unperturbed case, the gauge transformation
of a phase space function $X(K,E,\varphi,\pi)$ generated by the
diffeomorphism constraint acts as a Lie derivative
\begin{equation} \label{Lie}
\left\{X,D[\xi^a]\right\}={\mathcal L_{\vec\xi}X}
\end{equation}
along the vector field $\xi^a$. At the same time, the diffeomorphism
transformation of the perturbations is given by the $\xi$-part of
(\ref{Gauge_Trans_Explicit}), i.e.\ $\{\delta
X,D[\xi^a]\}=\delta_{[0,\xi]}\delta X$, whereas the barred quantities
remain intact at the linear level $\{\bar X,D[\xi^a]\}=O(2)$. For a
scalar field, these Poisson brackets can directly be identified as Lie
derivatives of background quantities and perturbations:
\[
{\mathcal L}_{\vec\xi}\bar \varphi=\xi^a\D_a\bar\varphi=0\quad,\quad
{\mathcal L}_{\vec\xi}\delta \varphi =\xi^a\D_a\delta\varphi=O(2)\,.
\]
However, if one computes the Lie derivatives of perturbative terms of
a tensorial object, or even of a scalar of non-zero density weight,
one can notice that they do not coincide with the $\xi$-part of
(\ref{Gauge_Trans_Explicit}). For instance,
\[
{\mathcal L}_{\vec\xi}\bar
\pi=\xi^a\D_a\bar\pi+\bar\pi\D_a\xi^a=\bar\pi\D_a\xi^a\quad,\quad
{\mathcal L}_{\vec\xi}\delta
\pi=\xi^a\D_a\delta\pi+\delta\pi\D_a\xi^a=O(2)
\]
while
\[
 \delta_{[0,\xi]}\bar{\pi}=0 \quad,\quad \delta_{[0,\xi]}\delta\pi =
\bar{\pi}\Delta\xi= \bar{\pi} \partial_a\xi^a
\]
for $\xi^a=\partial^a\xi$.  Similar discrepancies occur for the triad
and extrinsic curvature.

Nevertheless, gauge transformations are related to the Lie
derivative. As shown in App.~A of \cite{ConstraintAlgebra}, up to a
higher perturbative order we have
\be\label{Lie_Pert}
\{\bar
X,D[\xi^a]\}=[\l_{\vec\xi} X ]^{(0)}, \quad
\{\delta X,D^{(2)}[\xi^a]\}=[\l_{\vec\xi} X]^{(1)}
\ee
for the background and perturbed parts of $X$ respectively.  (As
explained in more detail in \cite{ConstraintAlgebra}, the second
Poisson bracket takes into account the fact that perturbations $\delta
X$ do not contain zero modes.)
These two equations above are consistent with the Lie
derivative of the full variable $X$
\[
\{\bar X+\delta X,D[\xi^a]\}=\l_{\vec\xi}  (\bar X+\delta X).
\]

Individual terms in this expansion, when computed either via
(\ref{Lie}) or via (\ref{Gauge_Trans}), do not agree in general
because $({\cal L}_{\vec{\xi}}X)^{(0)}$ may not equal ${\cal
L}_{\vec{\xi}}(X^{(0)})= {\cal L}_{\vec{\xi}}\bar{X}$: While the
$\xi$-gauge transformation of an unperturbed variable is equivalent to
taking a Lie derivative, the perturbation procedure breaks this
equivalence. For instance, the Lie derivative of a background quantity
(being linear) contributes to the gauge-transformation of its
perturbation, not of the background quantity itself. At the same time,
the Lie derivative of a linear perturbation (being at least a
quadratic quantity and hence neglected here) contributes to the
back-reaction on the background.  Their zero- and first-order parts of
the Lie derivative do contribute to the diffeomorphism transformation,
with combined contribution equal to the diffeomorphism transformation
of the full variable, although in a rather mixed way. When combined to
$\bar{X}+\delta X$, conventional transformations are obtained.


\subsubsection{Transformation of the lapse function and shift vector}

In the covariant formulation, the lapse function and shift vector are
merely components of the space-time metric, and hence subject to
coordinate (gauge) transformations in the same way as any other metric
component. In the Hamiltonian framework, on the other hand, lapse and
shift act as Lagrange multipliers and are not phase space
variables. Therefore, unlike e.g.\ triad components, their gauge
transformations cannot be directly obtained as Poisson brackets
(\ref{Gauge_Trans}) (which would always give zero) with the
gauge-generating constraints. Nevertheless, there exists an indirect
procedure.\footnote{Alternatively, Poisson brackets can be defined on
an extended phase space which also includes the Lagrange multipliers
\cite{LapseGauge}.}

A coordinate change causes a change in the space-time foliation by
spatial slices, and thus the induced $E^a_i$ and $K_a^i$ change
according to their gauge transformations. Since the slicing is
determined by lapse and shift they, too, must change. Lapse and shift
not only determine the slicing but also, as in (\ref{deltaEdot}) and
(\ref{phibardot}), equations of motion which triad and extrinsic
curvature have to satisfy as one moves from one slice to the next.
Consistency of equations of motion for the gauge-transformed canonical
variables thus requires certain transformations of the lapse and
shift: They have to change such that they generate the correct
equations of motion for the transformed $E^a_i$ and $K_a^i$, on which
a canonical gauge transformation has been applied. In this way, gauge
transformations for $N$ and $N^a$ result unambiguously, even though
these are not phase space variables.

Hamilton's equations such as (\ref{deltaEdot}) have a time derivative
of a phase space variable on their left hand side, while the right
hand side depends on phase space variables and the Lagrange
multipliers. Thus performing a gauge transformation on both sides of
the equations one can obtain the transformations of $\delta N$ and
$\delta N^a$.  The non-trivial part of this recipe is the gauge
transformation of the left hand side, as a time derivative of a phase
space variable is not itself a phase space variable, and hence its Poisson
bracket with constraints is not defined. Furthermore, a gauge
transformation does not, in general, commute with taking time
derivatives, that is the gauge transformation of a time derivative is
not merely given by the time derivative of the gauge transformation.

Nevertheless, the gauge transformation of a time derivative
can be computed with the help of
\begin{lemma}\label{GaugeDot}
For an arbitrary linear phase space function $\delta X$, the
commutator between its gauge transformation and its time
derivative is given (up to second order terms) by a single
gauge transformation
\begin{equation}\label{CommGaugeDot}
\delta_{[\xi_0,\xi]}(\delta \dot
X)-\left(\delta_{[\xi_0,\xi]}\delta
X\right)\dot{}=\delta_{[0,\xi_0]}\delta X
\end{equation}
\end{lemma}
\begin{proof}
Using the definition of the gauge transformation (\ref{Gauge_Trans})
and time derivative (\ref{deltaEdot}) via the Poisson bracket, by
virtue of the Jacobi identity we obtain
\begin{eqnarray}\label{ProofGaugeDot}
\delta_{[\xi_0,\xi]}(\delta \dot
X)-\left(\delta_{[\xi_0,\xi]}\delta
X\right)\dot{}&=&\left\{\{\delta X,H[N]+D[N^a]\},H[\bar N
\xi_0]+D[\D^a \xi]\right\}\nonumber\\&&-\left\{\{\delta X,H[\bar N
\xi_0]+D[\D^a \xi]\},H[N]+D[N^a]\}\right\}\nonumber\\
&=&\left\{\delta X, \{H[N]+D[N^a],H[\bar N \xi_0]+D[\D^a
\xi]\}\right\}\,.
\end{eqnarray}
The inner Poisson bracket can be computed using the constraint
algebra. In perturbative form, we have \cite{ConstraintAlgebra}
\bq
\label{HHD}\left\{H[N]+D[N^a],H[\bar N \xi_0]+D[\D^a
\xi]\right\}&=&D\left[\f{1}{\bar{p}}\left(N\D^a(\bar N \xi_0)-\bar
N\xi_0\D^a N\right)\right]-H\left[\D^a\xi\D_a N\right]\nonumber\\
&&+H\left[N^a\D_a(\bar N\xi_0)\right]+D\left[\D^c\xi \D_c N^a -
N^c\D_c\D^a\xi\right]\,. \eq
Most of these constraint terms are at least of second order, and
constraints whose Lagrange multiplier is quadratic will not affect the
(leading) linear part of the gauge transformation of $\delta
X$. Therefore the only relevant contribution comes from the first part
of the first term, $D\left[\f{\bar{N}}{\bar{p}}\D^a (\bar
N\xi_0)\right]$, which is equivalent to a single diffeomorphism
transformation with the (linear part of the) shift vector given by
\[
\f{\bar{N}}{\bar{p}}\D^a(\bar N\xi_0)=\D^a\xi_0+O(2)\,.
\]
The commutator then reads
\begin{eqnarray}
\delta_{[\xi_0,\xi]}(\delta \dot
X)-\left(\delta_{[\xi_0,\xi]}\delta X\right)\dot{}=\left\{\delta
X,D[\D^a\xi_0]\right\}\equiv \delta_{[0,\xi_0]}\delta
X\,.\nonumber
\end{eqnarray}
which is (\ref{CommGaugeDot}).
\end{proof}
The last equation implies that diffeomorphism invariant canonical
variables do have commuting time derivative and gauge
transformation. Moreover, the leading diffeomorphism term
originates from the Poisson bracket of the two Hamiltonian constraints
on the left hand side of (\ref{HHD}). Therefore, taking a time derivative
commutes (up to quadratic terms) with a diffeomorphism
transformation. This can also be seen from the absence of $\xi$ on the
right hand side of (\ref{CommGaugeDot}). In other words, gauge
transformations which do not involve the Hamiltonian constraint
(i.e. such that $\xi_0=0$) commute with taking time derivatives.

For later convenience we write out gauge transformed time
derivatives for a number of phase space variables:
\begin{eqnarray}\label{GaugeDotPhVar}
\delta_{[\xi_0,\xi]}(\dot \psi)&=&\left(\delta_{[\xi_0,\xi]}\psi\right)\dot{}\nonumber\\
\delta_{[\xi_0,\xi]}(\dot E)&=&\left(\delta_{[\xi_0,\xi]} E\right)\dot{}+\xi_0\nonumber\\
\delta_{[\xi_0,\xi]}(\delta \dot \varphi)&=&\left(\delta_{[\xi_0,\xi]}\delta \varphi\right)\dot{}\nonumber\\
\delta_{[\xi_0,\xi]}(\delta \dot
\pi)&=&\left(\delta_{[\xi_0,\xi]}\delta
\pi\right)\dot{}+\pi\Delta\xi_0,
\end{eqnarray}
where the gauge transformations of the triad components, following
from (\ref{Gauge_Trans_Explicit}) by comparison with (\ref{Triad}),
are given by
\begin{eqnarray}\label{GaugeEpsi}
\delta_{[\xi_0,\xi]}\psi=-\H\xi_0\quad , \quad
\delta_{[\xi_0,\xi]}E=\xi.
\end{eqnarray}
We now have all the ingredients to obtain the transformations of lapse
and shift perturbations which (for the scalar mode and in conformal
time) are expressed as in (\ref{PertLapseShift}). Writing the
perturbation of extrinsic curvature using the equation of motion
\[ \delta \dot E_i^a
\equiv \{\delta E_i^a, H^{(2)}[\delta N]+D^{(2)}[\delta     N^a]\}
\]
for the perturbed triad and the expression
$\H=\dot{a}/a=\dot{\bar{p}}/2\bar{p}$ for the Hubble parameter, we
obtain
\begin{equation}\label{deltaK_Cl}
\delta K^i_a=-\delta^i_a\left[\dot \psi+\H(\psi
+\phi)\right]+\D_a\D^i\left[\H E-(B-\dot E)\right].
\end{equation}
Performing a gauge transformation of the left hand side according to
(\ref{Gauge_Trans_Explicit}) and comparing to the gauge transformation
of the right hand side using (\ref{GaugeDotPhVar}) and
(\ref{GaugeEpsi}) yields the desired transformation of the Lagrange
multipliers. Specifically, the diagonal part provides the transformed
lapse perturbation
\be\label{GT_phi_Cl}
\delta_{[\xi_0,\xi]}\phi=\dot\xi_0+\H\xi_0\,,
\ee
whereas the
off-diagonal part implies \be\label{GT_BEdot_Cl}
\delta_{[\xi_0,\xi]}(B-\dot E)=-\xi_0.%
\ee
Thus, the transformations of lapse and shift are indeed determined by
gauge transformations of the phase space variables through the dynamical
equations of motion.

Gauge invariant combinations are then obtained from the gauge
transformations, which reproduces the Bardeen variables
(\ref{Bardeen}).  Note that, as follows from (\ref{GT_BEdot_Cl}), the
quantity $B-\dot E$ is diffeomorphism invariant. According to
Lemma~\ref{GaugeDot}, the gauge transformation of its time derivative
is then given simply by the time derivative of its gauge
transformation. Consequently, the last term in the $\Phi$-equation
(\ref{Bardeen}) gauge-transforms as
\[
\delta_{[\xi_0,\xi]}\left\{(B-\dot
E)^{\dot{}}\right\}=\left\{\delta_{[\xi_0,\xi]}(B-\dot
E)\right\}^{\dot{}}=-\dot\xi_0,
\]
which along with the second term of $\Phi$ compensates the
gauge-dependence of $\phi$.

We have determined gauge transformations of lapse and shift by making
sure that the form of equations of motion for $E^a_i$ and $K_a^i$ is
invariant. Here, we used the basic fact that changing the lapse
function and shift vector leads to a different space-time metric
decomposition and hence a different form of the triad and extrinsic
curvature as well as their evolution. A different slicing of space-time
also affects the matter variables, e.g. the definition of the field
momentum. Gauge transformations of the matter variables will induce
transformations of $N$ and $N^a$ through Hamilton's equations for
$\varphi$ and $\pi$ as they did for gravitational phase space
variables. We must therefore ensure that the transformations of the
lapse and shift generated in this way be consistent with those
obtained in (\ref{GT_phi_Cl}) and (\ref{GT_BEdot_Cl}). Taking the
gauge transformation of the left hand side of the equation of motion
\begin{eqnarray}
\delta \dot \varphi &\equiv& \{\delta \varphi,
H^{(2)}[N]+D^{(2)}[N^a]\}=\f{\bar N}{\bar{p}^{3/2}}\left(\delta
\pi-\bar{\pi}\f{(\delta E_i^a
\delta^i_a)}{2\bar{p}}\right)+\f{\delta N}{\bar{p}^{3/2}}\pi\nonumber\\
&=&\f{\delta\pi}{\bar{p}}+\f{\bar{\pi}}{\bar{p}}\left(3\psi-\Delta
E+\phi\right)\label{Pert_EoM_varphi}
\end{eqnarray}
according to (\ref{GaugeDotPhVar})
and (\ref{Gauge_Trans_Explicit}), and comparing it with the
transformation of the right hand side, yields the transformation of
the lapse perturbation, $\delta_{[\xi_0,\xi]}\phi=\dot\xi_0+\H\xi_0$
which agrees with Eq.~(\ref{GT_phi_Cl}). Repeating the procedure for
the momentum equation
\begin{eqnarray}
\delta \dot \pi &\equiv& \{\delta \pi,
H^{(2)}[N]+D^{(2)}[N^a]\}=\bar
N\left[\sqrt{\bar{p}}\Delta\delta\varphi-{\bar{p}^{3/2}}
\left(V_{,\varphi\varphi}\delta
\varphi+V_{,\varphi}\f{(\delta E_i^a
\delta^i_a)}{2\bar{p}}\right)\right]\nonumber\\
&=&\bar{p}\Delta\delta\varphi-\bar{p}^2\left(V_{,\varphi\varphi}\delta
\varphi-V_{,\varphi}(3\psi-\Delta
E+\phi)\right)+\bar{\pi}\Delta B\label{Pert_EoM_pi}\,,
\end{eqnarray}
results in $\delta_{[\xi_0,\xi]}B=\dot\xi$. The latter along with
(\ref{GaugeDotPhVar}) and (\ref{GaugeEpsi}), implying
\[
\delta_{[\xi_0,\xi]}\dot
E=\left(\delta_{[\xi_0,\xi]}E\right)\dot{}+\xi_0=\dot\xi+\xi_0\,,
\]
reproduces the correct gauge transformation (\ref{GT_BEdot_Cl}).  Thus
a fixed transformation of lapse and shift provides the correct gauge
transformation of both gravity and matter phase space variables. This
is a further consistency property ensured by the first class nature of
constraints.

For matter fields, the gauge invariant density and scalar field
perturbations are
\be\label{GI_varphi_Cl}
\delta\rho^{\rm GI}=\delta\rho + \dot{\bar{\rho}_{_\varphi}} \,  (B-\dot E) \, ;
\quad 
\delta\varphi^{\rm GI}=\delta\varphi + \dot{\bar{\varphi}} \, (B-\dot E) \, .
\ee

It is convenient for cosmological applications to introduce
the gauge-invariant quantities:
\begin{equation}
{\cal R} = \Psi+\H \left(\frac{\delta{\varphi}^{{\rm
GI}}}{\dot{\bar{\varphi}}}\right)
=\psi+\H\left(\frac{\delta{\varphi}}{\dot{\bar{\varphi}}}\right)
\,\label{GI_CurvPert1}
\end{equation}
\begin{equation}
- \zeta = \Psi +  \H \left(\frac{\delta{\rho}^{{\rm
GI}}}{\dot{\bar{\rho}}_{_{_\varphi}}}\right)
=\psi+\H \frac{\delta\rho}{\dot{\bar{\rho}}_{_\varphi}}
\end{equation}
and
\begin{equation}
{\mathcal R}_2 = \Phi-\H\left(\frac{\delta{\varphi}^{{\rm
GI}}}{\dot{\bar{\varphi}}}\right)
-{\left(\frac{\delta{\varphi}^{{\rm
GI}}}{\dot{\bar{\varphi}}}\right)}^{\dot{}} ~=~ \phi-\H
\left(\frac{\delta\varphi}{\dot{\bar{\varphi}}}\right)
-{\left(\frac{\delta\varphi}{\dot{\bar{\varphi}}}\right)}^{\dot{}}
\,\label{GI_CurvPert2}
\end{equation}
The following points are worth noting regarding the above three gauge-invariant
quantities:
(i) ${\cal R}$ provides information about the nature of the
long wavelength perturbations i.e.\ when the perturbations have left the
Hubble radius. More precisely, $\dot{\cal R}$ vanishes in the long wavelength
limit if the perturbations are adiabatic \cite{CosmoPert,PowerLargeScales}.
(ii) $\zeta$ refers to the 3-curvature perturbations on uniform density
hypersurfaces. As ${\cal R}$, $\zeta$ is also conserved in the
large scales and quantifies the large angular scale temperature anisotropies
in the cosmic microwave background. In the slow-roll limit, the two
gauge-invariant quantities are identical and either of them can be used
to quantify the primordial perturbations.
(iii) Unlike $\zeta$ and ${\cal R}$, the Bardeen potential $\Phi$
evolves in time from Hubble exit until the re-entry during
matter/radiation era. More precisely, during inflation, at
the super-Hubble scales,
\begin{equation}
\label{Phi-zeta}
\Phi \simeq \epsilon_{_{sr}} \zeta \quad \Longrightarrow \quad \Phi \simeq \epsilon_{_{sr}} |A|
\end{equation}
where $\epsilon_{_{sr}}$ is the slow-roll parameter which is much less
than unity and $|A|$ is the value of $\zeta$ at the super-Hubble
scales. However, at horizon re-entry, $\Phi \sim \zeta$.
(iv) ${\cal R}$ (and also $\zeta$) is linearly related to the
Mukhanov-Sasaki variable $Q$ which is useful for studying the
quantization of perturbations.
(v) Unlike ${\cal R}$ and $\zeta$, ${\cal R}_2$ is not often used in
the study of cosmological perturbations.
(vi) In the quantum corrected version of the
perturbation equations, all the above quantities acquire non-trivial
quantum corrections which we will discuss in the rest of the paper.
This also affects the conservation of power at large scales.

For the classical gauge transformations we thus produce the well-known
gauge invariant quantities (\ref{Bardeen}), but we have now done so in a
way which is entirely canonical. These methods therefore generalize
directly to the case of equations and gauge transformations which are
corrected by effects from quantum gravity even in quantum regimes
where no underlying smooth space-time picture exists.

\subsection{Inclusion of quantum corrections}

With quantum gravity effects, both the equations of motion and gauge
transformations are governed by the quantum corrected Hamiltonian
constraint $H^Q$ in (\ref{H2Q}), including all the counter-terms and
the diffeomorphism constraint (\ref{D2gauge}) which remains
unaffected. Recall also that the terms $\alpha^\prime p$, $\nu^\prime
p$ as well as the counter-term functions are leading order quantum
effects (not to be confused with perturbative order). From now on we
neglect all higher order quantum corrections, such as
$(\alpha^\prime)^2 p^2$, $f^2$, $\alpha^\prime p f$ etc. Moreover,
also in the background equations (\ref{FriBG}), (\ref{RayBG}) and
(\ref{KGBG}) the quantum corrected terms must be used for consistency.

Gauge transformations are generated by the quantum corrected
constraints with the lapse function and shift vector parameterized
by the infinitesimal vector field $\xi^\mu=(\xi_0,\D^a\xi)$ as
\be\label{LapseShift_Def} \delta N = \bar N\xi_0, \quad \delta N^a
= \D^a \xi.
\ee
The gauge transformation of the triad perturbation taking into
account the counter-terms can be computed using the Poisson bracket
\be \delta_{[\xi_0,\xi]} E^a_i\equiv \{\delta E_i^a,H^{Q(2)}[\bar
N\xi_0]+D^{(2)}[\D^a\xi]\}=2\bar{\alpha} \bar{k} \bar{p}
\xi_0 (1+f)\delta_i^a + \bar{p}
(\delta_i^a \Delta\xi-\D^a\D_i\xi)\nonumber \ee
from which the transformations for $\psi$ and $E$ follow
\be\label{GT_psiE}%
\delta_{[\xi_0,\xi]} \psi=-\H\xi_0(1+f) , \quad \delta_{[\xi_0,\xi]} E=\xi.
\ee
The gauge transformed extrinsic curvature
yields%
\be\label{GT_K}
\delta_{[\xi_0,\xi]} (\alpha K_a^i) =
\D^i\D_a(\H\xi+\bar{\alpha}^2\xi_0)-
\left[\f{1}{2}\H^2\xi_0(1+g)-4\pi G\xi_0\bar{\alpha}(\bar{p}V(\bar\varphi)
-\f{\dot{\bar\varphi}^2}{2\bar\nu}(1+f_2) )\right]\delta^i_a. %
\ee%
We compare this equation with the gauge transformation of
(\ref{deltaK}), eliminating the potential term using the
background Raychaudhuri equation (\ref{RayBG}). By virtue of the
anomaly-freedom condition (\ref{f1}) along with
\cite{ConstraintAlgebra}
\begin{equation}\label{f1_f2}
f_2=2f_1 \,,
\end{equation}
the off-diagonal part results in
\be\label{GT_BEdot} \delta_{[\xi_0,\xi]}(B-\dot
E)=-\bar{\alpha}^2\xi_0\,, \ee
whereas the diagonal part yields
\be\label{GT_phi}
\delta_{[\xi_0,\xi]}\phi=\dot\xi_0+\H\xi_0\left(1+2f^\prime
\bar{p}+\f{\bar{\alpha}^\prime \bar p}{\bar{\alpha}}\right)\,. \ee
Remarkably, the anomaly cancellation condition (\ref{fgh}) implies
that the last two terms inside the parenthesis mutually cancel,
hence quantum corrected transformation of the lapse perturbation
is equivalent to the classical one (\ref{GT_phi_Cl}).
Finally the matter perturbation transforms
as
\be\label{GT_varphi}
\delta_{[\xi_0,\xi]}\delta\varphi=\dot{\bar\varphi}(1+f_1)\xi_0\,.
\ee
The four metric perturbations can be combined into two gauge invariant
quantum corrected potentials
\bq\label{BardeensPot}
\Psi&=&\psi-\H(1+f)\f{B-\dot E}{\bar{\alpha}^2}\\
\Phi&=&\phi+\left(\f{B-\dot E}{\bar{\alpha}^2}\right)^{\dot{}}+\H\f{B-\dot
E}{\bar{\alpha}^2}
\,.\nonumber \eq
Similarly, the gauge invariant matter variables are
\be\label{GI_varphi}
\delta\varphi^{\rm GI}=\delta\varphi+\dot{\bar{\varphi}}(1+f_1)\f{B-\dot
E}{\bar{\alpha}^2} \,
\ee
Note that when omitting the quantum corrections in
Eqs.~(\ref{BardeensPot}) and (\ref{GI_varphi}) one recovers the
classical results (\ref{Bardeen}).  {}From the corrected gauge
invariant expressions, one can directly see that the combination
\begin{equation}
 {\cal R} = \psi+\frac{\H}{\dot{\bar{\varphi}}} \frac{1+f}{1+f_1}
 \delta\varphi=\Psi+\frac{\H}{\dot{\bar{\varphi}}} \frac{1+f}{1+f_1}
 \delta\varphi^{\rm GI}\,,
\end{equation}
which does not refer to the non-trace perturbations $E$ and $B$, is
gauge invariant. Also the explicit $\bar{\alpha}$-dependence drops out,
showing that this particular perturbation is quantum corrected only
because we were required to include counter-terms $f$ and $f_1$ in
addition to the primary correction function $\bar{\alpha}$. Similarly
%
the other curvature perturbation is
\begin{equation}
 {\cal R}_2= \phi-\H\frac{\delta\varphi}{\dot{\bar{\varphi}}} \frac{1}{1+f_1}
- \left(\frac{\delta\varphi}{\dot{\bar{\varphi}}(1+f_1)}\right)^.=
\Phi-\H\frac{\delta\varphi^{\rm GI}}{\dot{\bar{\varphi}}}
\frac{1}{1+f_1} - \left(\frac{\delta\varphi^{\rm
GI}}{\dot{\bar{\varphi}}(1+f_1)}\right)^.
\end{equation}
which does not refer to $\bar{\alpha}$, either.
(A generalization of the perturbation $\zeta$ to quantum corrected
equations requires a general derivation of $\delta\rho^{{\rm GI}}$,
which will be done elsewhere.)

\section{Gauge invariant equations of motion}

We can now formulate the perturbed equations of
motion purely in terms of the gauge invariant variables derived in
the previous section. The following auxiliary relations will be
useful:
\begin{eqnarray}\label{GaugeAux}
\delta E^a_i&=&-2\bar{p}\Psi\delta_i^a \nonumber\\
&&-2\H \bar{p}(1+f)\f{B-\dot E}{\bar{\alpha}^2} \delta_i^a +\bar{p}
(\delta_i^a\Delta-\D^a\D_i)E \\
 \bar{\alpha} \delta K_a^i&=&-\delta_a^i\left(\dot\Psi+
\H\left(\Psi+\Phi(1+f)\right)\right)\nonumber\\
&&-\delta_a^i\f{B-\dot
E}{\bar{\alpha}^2}\left(\dot\H(1+f)-\H^2\f{\bar{\alpha}^\prime
\bar{p}}{\bar{\alpha}}\right)+\D_a \D^i\left(\H E - (B-\dot
E)\right)\nonumber
\end{eqnarray}
where the first line of each equation contains only gauge
invariant terms.

\subsection{Diffeomorphism constraint equation}
Varying the diffeomorphism constraint (\ref{DiffeoConstraint}) with
respect to the shift perturbation yields the diffeomorphism constraint
equation (the space-time component of Einstein's equation):
\begin{equation}
0=8\pi G\alpha\f{\delta D[N^c]}{\delta (\delta
N^c)}=\bar{p}\left(\D_c(\bar{\alpha}\delta
K_a^i\delta_i^a)-\D_k(\bar{\alpha}\delta K_c^k)\right)-\bar{\alpha}
\bar{k}\D_d\delta E^d_k \delta_c^k+8\pi G\bar{\alpha} \bar{\pi}
\D_c\delta\varphi\,.
\end{equation}
Using the gauge invariant variables defined in (\ref{BardeensPot})
and (\ref{GI_varphi}), it can be rewritten as
\begin{equation}\label{DiffEq}
\D_c\left[\dot\Psi+\H(1+f)\Phi-4\pi
G\f{\bar{\alpha}}{\bar{\nu}}\dot\varphi\delta\varphi^{\rm GI}\right]+({\rm
gauge \, \,terms})=0,
\end{equation}
where the `gauge terms' are
\begin{equation} \label{GaugeTerms}
2\f{B-\dot
E}{\bar{\alpha}}\D_c\left[-\dot\H(1+f)+\H^2\left(1+f+\f{\bar{\alpha}^\prime
\bar{p}}{\bar{\alpha}}\right)-4\pi G
\f{\bar{\alpha}}{\bar{\nu}}\dot\varphi^2(1+f_1)\right]
\end{equation}
After eliminating the $\dot \H$-term using the background
Raychaudhuri equation (\ref{RayBG}) the expression inside the
square brackets of (\ref{GaugeTerms}) becomes
\[
4\pi G \f{\bar{\alpha}}{\bar{\nu}}\dot{\bar{\varphi}}^2
\left(f-f_1-\f{\bar{\nu}^\prime
\bar{p}}{3\bar{\nu}}\right)
\]
which is proportional to one of the anomaly-freedom conditions
(\ref{f1}). Thus all gauge dependent terms vanish and the
diffeomorphism constraint equation takes the form
\begin{equation}\label{DiffEqFinal}
\D_c\left(\dot\Psi+\H(1+f)\Phi\right)=4\pi
G\f{\bar{\alpha}}{\bar{\nu}}\dot\varphi \D_c\delta\varphi^{\rm GI}\,.
\end{equation}
Note that classically the right hand side is nothing but the
gauge invariant space-time component of the perturbed matter
stress-energy tensor $-4\pi G a^2 \delta T_S^T$.

\subsection{Hamiltonian constraint equation}

As seen explicitly for the diffeomorphism constraint, for constraints
which are part of a closed system, i.e.\ which result from an
anomaly-free quantization, gauge invariance of the equations of motion
is guaranteed and showing that the gauge dependent terms of a given
equation do vanish is in general a rather tedious, although
straightforward, exercise. We leave out such explicit demonstrations
in this and the following sections.

The Hamiltonian constraint equation is obtained by variation with
respect to the lapse perturbation:
\begin{eqnarray}
\f{\delta H^Q[N]}{\delta (\delta N)}&=&\f{1}{16\pi G}\left[-4\bar{\alpha}
\bar{k}\sqrt{\bar{p}}(1+f)\delta K_a^i \delta_i^a-\f{\bar{\alpha}
\bar{k}^2}{\sqrt{\bar{p}}}(1+g)\delta
E_i^a\delta_a^i+\f{2\bar{\alpha}}{\sqrt{\bar{p}}}\D_a\D^i\delta
E^i_a\right]\nonumber\\
&+&\f{\bar{\nu}\bar{\pi}\delta\pi}{\bar{p}^{3/2}}(1+f_1)-
\left(\f{\bar{\nu}\bar{\pi}^2}{2\bar{p}^{3/2}}(1+f_2)
-\bar{p}^{3/2}V(\varphi)\right)\f{\delta E_i^a
\delta^i_a}{2\bar{p}}+\bar{p}^{3/2}V_{,\varphi}(\bar{\varphi})(1+f_3)
\delta\varphi\nonumber\\
&=&0\, ,
\end{eqnarray}
where
\begin{equation}
\label{f3}
f_3 = \frac{3}{2 \bar{p}^{3/2}} \int d\bar{p} \, \bar{p}^{1/2} \, f
\end{equation}
can be obtained from (\ref{f3diff}).

Multiplying both sides by $\bar{\alpha}/\sqrt{\bar{p}}$ again allows
one to replace the background extrinsic curvature with the Hubble
rate.  Then eliminating the field momentum and its perturbation using
(\ref{deltapi}) and the auxiliary expressions (\ref{GaugeAux}) along
with anomaly-freedom conditions of \cite{ConstraintAlgebra} to reduce
the number of counter-term functions, one arrives at the gauge
invariant Hamilton constraint equation (or perturbed Friedmann
equation)
\begin{eqnarray}\label{HamEqFinal}
\Delta(\bar{\alpha}^2
\Phi)-3\H(1+f)\left[\dot\Psi+\H\Phi(1+f)\right]&=&4\pi
G\f{\bar{\alpha}}{\bar{\nu}}(1+f_3)\left[\dot{\bar{\varphi}}
\delta\dot\varphi^{\rm
GI}-\dot{\bar{\varphi}}^2(1+f_1)\Phi\nonumber\right.\\
&&+\left.\bar{\nu} \bar{p} V_{,\varphi}(\bar{\varphi})
\delta\varphi^{\rm GI}\right]
\end{eqnarray}
Again, the right hand side is nothing but the time-time component of
the perturbed stress-energy tensor, which now includes quantum
corrections.

\subsection{Hamilton's equations}
As mentioned before, each pair of Hamilton's equations for
configuration variables and momenta can be combined into a single
second order equation. We illustrate the procedure starting with the
matter field. The time derivative of the momentum perturbation can be
first computed using the Poisson bracket
\begin{eqnarray}
\delta\dot\pi&=&\{\delta\pi,H[N]+D[N^a]\}=\bar{\pi} \D_a \delta N^a -
\delta N \bar{p}^{3/2} V_{,\varphi} (\bar{\varphi}) (1+f_3)\\&&+\bar N
\left[\sqrt{\bar{p}}\bar{\sigma}(1+g_5)\Delta\delta\varphi-
\bar{p}^{3/2}V_{,\varphi\varphi}(\bar{\varphi})(1+g_6)\delta\varphi-
\bar{p}^{3/2}V_{,\varphi}(\bar{\varphi})\f{\delta
E_i^a\delta_a^i}{2\bar{p}}\right],\nonumber
\end{eqnarray}
and then compared to the time derivative of the right hand side of
Eq.~(\ref{deltapi}), which will include second time derivatives of
the scalar field. With the help of the background equations and
anomaly cancellations conditions, the Klein-Gordon equation can be
cast in the gauge invariant form
\bq\label{GI_KG}%
\delta \ddot \varphi^{\rm GI}\!\!\! &+&\!\!\! 2 \H \delta \dot
\varphi^{\rm GI} \left(1 - \f{\bar{\nu}^\prime \bar{p}}{\bar{\nu}}-g_1^\prime
\bar{p}\right)-\bar{\nu} \bar{\sigma} (1-f_3)\Delta \delta \varphi^{\rm GI}
+ \bar{\nu} \bar{p}
V,_{\varphi\varphi}(\bar{\varphi})\delta \varphi^{\rm GI}
\\\!\!\!&+&\!\!\! 2\bar{\nu} \bar{p}
V,_{\varphi}(\bar{\varphi})(1+f_1) \Phi
-\dot{\bar{\varphi}}\left[(1+f_1)\dot\Phi+
3(1+g_1)\dot\Psi\right]-2\H\dot{\bar{\varphi}}
(f_3^\prime \bar{p})\Phi=0\,.\nonumber
\eq

Similarly one arrives at the spatial components of Einstein's
equation. Taking the time derivative of Eq.~(\ref{deltaK}) and noting
that
\[
\left(\bar{\alpha}\delta K_a^i\right)^{\dot{}}\equiv \bar{\alpha}\delta\dot
K_a^i+\delta K_a^i\dot{\bar{\alpha}}=\alpha\delta\dot K_a^i+2\H \bar{\alpha}
\delta K_a^i \left(\f{\bar{\alpha}^\prime \bar{p}}{\bar{\alpha}}\right),
\]
one can substitute the time derivative of the extrinsic curvature
perturbation using the Poisson bracket
\[
\delta \dot K_a^i=\{\delta K_a^i, H[N]+D[N^b]\}.
\]
The resulting expression will contain second order correction
functions $\alpha^{(2)}$ and $\nu^{(2)}$ related to the background
ones, $\bar{\alpha}$ and $\bar{\nu}$, by the conditions (\ref{alpha2})
for anomaly freedom.

Taking the trace of each equation and substituting it back into the
left hand side, one obtains
\begin{equation} \label{alpha2simple}
\frac{\partial\alpha^{(2)}}{\partial(\delta E^a_i)} =
\f{\bar{\alpha}^\prime \delta E^c_j}{6\bar{p}}
\left(\delta^j_c\delta_a^i-2\delta_c^i\delta_a^j\right), \quad
\frac{\partial\nu^{(2)}}{\partial(\delta E^a_i)} = \f{\bar{\nu}^\prime
\delta E^c_j}{6\bar{p}}
\left(\delta^j_c\delta_a^i-2\delta_c^i\delta_a^j\right)
\end{equation}
whose gauge invariant parts read
\begin{equation} \label{GIalpha2}
\left(\frac{\partial\alpha^{(2)}}{\partial(\delta
E^a_i)}\right)^{\rm GI} =-\f{\bar{\alpha}^\prime}{3}\Psi\delta^i_a,
\quad \left(\frac{\partial\nu^{(2)}}{\partial(\delta
E^a_i)}\right)^{\rm GI} =-\f{\bar{\nu}^\prime}{3}\Psi\delta^i_a.
\end{equation}
The combined second order equation naturally decouples into two
independent equations by taking diagonal and off-diagonal
terms. After a tedious but rather straightforward computation, taking
into account the background equations of motion together with the
anomaly-freedom conditions, the former equation takes the form
\begin{eqnarray}\label{GI_Diag}
\ddot\Psi+\H\left[2\dot\Psi\left(1-\f{\bar{\alpha}^\prime
\bar{p}}{\bar{\alpha}}\right)+\dot\Phi(1+f)\right]
&+& \left[ \dot\H + 2 \H^2\left(1+f^\prime \bar{p} -
\f{\bar{\alpha}^\prime \bar{p}}{\bar{\alpha}}\right)\right]\Phi(1+f)
\nonumber \\
&=& 4\pi
G\f{\bar{\alpha}}{\bar{\nu}}\left[\dot\varphi\delta\dot\varphi^{\rm
GI}-\bar{p}\bar{\nu} V_{,\varphi}(\bar{\varphi})\delta\varphi^{\rm GI}
 \right]\,.
\end{eqnarray}
In the absence of anisotropic stress in the matter sector, which
is the case for the scalar field, the gauge invariant part of the
off-diagonal equation reads:
\begin{equation}\label{GI_OffDiag}
\D_a\D^i(\bar{\alpha}^2\left(\Phi-\Psi(1+h)\right))=0,
\end{equation}
which implies $\Phi=\Psi(1+h)$, replacing the classical relation
$\Phi=\Psi$. Here, $h$ is a counter-term correction which has to satisfy
\begin{equation} \label{h}
 h=-f+2\frac{\bar{\alpha}'\bar{p}}{\bar{\alpha}}\,.
\end{equation}

\section{Qualitative properties of the scalar perturbations}
\label{s:Qualitative}

In this section, we discuss salient properties of the scalar
perturbations including inverse triad corrections from loop quantum
gravity. For now, we do not compute the power-spectrum since this
would involve reducing the perturbation equations (\ref{DiffEqFinal}),
(\ref{HamEqFinal}), (\ref{GI_Diag}), (\ref{GI_KG}) into a single
differential equation in terms of the Mukhanov-Sasaki variable.
Instead, we will focus on aspects of the classical matter
perturbations in an effective quantum space-time. In particular, we
show that (i) the speed of scalar perturbations is less than unity and
can in fact be much smaller, and (ii) the scalar perturbations are not
purely adiabatic and have a small entropic contribution arising from
the quantum correction.

\subsection{Speed of perturbations}

Using the relation $\Phi = (1 + h) \Psi$, the perturbation equations
(\ref{DiffEqFinal}), (\ref{HamEqFinal}), (\ref{GI_Diag}) lead to
%
\begin{eqnarray}
\label{eq:dT00}
& & \bar{\alpha}^2 (1 + h) \Delta\Psi
- 3\H(1+f)\left[\dot\Psi+\H(1+f)(1 + h) \Psi \right]  \\
& & \qquad \qquad \qquad \qquad \qquad = 4\pi G\f{\bar{\alpha}}{\bar{\nu}}(1+f_3)
\left[\dot{\bar{\varphi}} \delta\dot\varphi^{\rm GI}
- \dot{\bar{\varphi}}^2(1+f_1)\Phi
+ \bar{\nu} \bar{p} V_{,\varphi}(\bar{\varphi})
\delta\varphi^{\rm GI}\right] \nonumber \\
\label{eq:dT0i}
& & \dot\Psi+\H(1+f)(1 + h) \Psi =
4\pi G\f{\bar{\alpha}}{\bar{\nu}}\dot\varphi \delta\varphi^{\rm GI} \\
\label{eq:dTii}
& &  \ddot\Psi+\H\left[2\left(1-\f{\bar{\alpha}^\prime
\bar{p}}{\bar{\alpha}}\right) \dot\Psi +(1+f) (1 + h) \dot\Psi \right]  \\
& & + \left[\dot\H + 2\H^2\left(1+f^\prime \bar{p} -
\f{\bar{\alpha}^\prime \bar{p}}{\bar{\alpha}}\right)
+  \frac{\dot{h}}{(1 + h)} \right] (1+f) (1 + h) \Psi \nonumber
= 4\pi G\f{\bar{\alpha}}{\bar{\nu}}\left[\dot\varphi\delta\dot\varphi^{\rm
GI}-\bar{p}\bar{\nu} V_{,\varphi}(\bar{\varphi})\delta\varphi^{\rm GI}\right]\,.
\end{eqnarray}
Dividing (\ref{eq:dT00}) by $1 + f_3$ and subtracting the result from
Eq.~(\ref{eq:dTii}) leads to
\begin{eqnarray}
\label{eq:temp1}
& & \ddot\Psi - \bar{\alpha}^2 \frac{1 + h}{1 + f_3} \Delta \Psi
+ \H \left[2 \left(1-\f{\bar{\alpha}^\prime
\bar{p}}{\bar{\alpha}}\right) +  (1 + f) (1 + h) + 3 \frac{1 + f}{1 + f_3} \right] \dot{\Psi}  \\
& & \left[
\left\{\H \frac{\dot{h}}{1 + h}  + \dot\H  + 2 \H^2 \left(1 + f' \bar{p} - \f{\bar{\alpha}^\prime
\bar{p}}{\bar{\alpha}} \right) \right\} (1 + f) (1 + h) +  \H^2 (1 + h)
\right. \nonumber \\
& & \left. \times
\left\{3 \frac{(1 + f)^2}{1 + f_3} - \f{1 + \bar{\alpha}^\prime\bar{p}/\bar{\alpha}}{1 - \bar{\nu}^{\prime}\bar{p}/\bar{\nu} } (1 + f_1) \right\} 
+  \dot\H \f{(1 + f_1) (1 + h)}{1 -\bar{\nu}^{\prime}\bar{p}/\bar{\nu} } \right] \Psi
= - 8 \pi G \bar{\alpha} \bar{p} (1 + f_3) V_{,\bar{\varphi}} \delta \varphi^{{\rm GI}}\,. \nonumber
\end{eqnarray}
Replacing $V_{,\bar{\varphi}}$ and $\delta \varphi$
using the Klein-Gordon equation (\ref{KGBG}) and the
diffeomorphism constraint (\ref{eq:dT0i}), we obtain a second-order
partial differential equation for $\Psi$:
{\small
\begin{eqnarray}
\label{eq:FinPhi}
& & \ddot\Psi - c_s^2 \Delta \Psi
+ \left[2 \H \left\{1 + \frac{3}{2} \frac{f -f_3}{1 + f_3}
+ 2 \frac{\bar{\nu}^{\prime}}{\bar{\nu}} \bar{p}
+  \frac{\bar{\alpha}^{\prime}}{\bar{\alpha}} \bar{p} f - \frac{f^2}{2} \right\}
- 2 \frac{\ddot{\bar\varphi}}{\dot{\bar\varphi}} \right] \dot\Psi  \\
& & \left[
\H \left\{\frac{\dot{h}}{1 + h}
-  2 \frac{\ddot{\bar\varphi}}{\dot{\bar\varphi}} \right\} (1 + f)
+ \dot\H \left\{\frac{1 + f_1}{[1 - \bar{\nu}' \bar{p}/(3\bar{\nu})]} + (1 + f) \right\}
+ \H^2 \left\{ 3 \frac{(1 + f)^2}{1 + f_3} \right.
\right. \nonumber \\
& & \left. \left.
+ 2 \left(1 + f^\prime \bar{p} - \f{\bar{\alpha}^\prime \bar{p}}{\bar{\alpha}} \right) (1 + f)
- \f{(1 + \bar{\alpha}^\prime\bar{p}/\bar{\alpha})}{(1 - \bar{\nu}^{\prime}\bar{p}/(3 \bar{\nu}))} ( 1 + f_1)
- 4 \left(1 - \frac{\bar{\nu}^{\prime}}{\bar{\nu}} \bar{p} \right) ( 1 + f)\right\}
\right](1 + h)\Psi = 0 \nonumber
\end{eqnarray}
}
where the second term shows the speed of perturbations
\begin{equation}
\label{eq:cs}
c_s^2 =  \bar{\alpha}^2 \frac{1 + h}{1 + f_3}\,.
\end{equation}
Eqs.~(\ref{eq:FinPhi}) and (\ref{eq:cs}) constitute one of the main
results of this paper, which has the
following implications:
\begin{enumerate}
\item[(i)] In the perturbative regime analyzed here, we have
  $\bar{\alpha}>1$ such that there is a danger of the speed of sound
  becoming super-luminal. However, as in the context of gravitational
  waves in \cite{tensor}, one must compare the propagation speed not
  with the classical speed of light (which is one) but with the
  physical speed of light of electromagnetic fields in the same
  effective quantum space-time. Since the Maxwell Hamiltonian is
  subject to quantum gravity corrections, too, \cite{QSDV,MaxwellEOS}
  the physical speed of light can differ from the classical one. In
  fact, for an anomaly-free coupling of the Maxwell field to gravity
  it must be larger than one by a factor which equals $\bar{\alpha}^2$
  as shown in \cite{tensor}. The speed of perturbations derived here
  is super-luminal compared to the physical speed of light only if the
  remaining factor $(1+h)/(1+f_3)$ in (\ref{eq:cs}) is larger than
  one. Using Eqs.~(\ref{h}) and (\ref{f3}) it is easy to show that, in
  the perturbative regime considered here (where $\bar{\alpha}'$ is
  negative), we have $(1+h)/(1+f_3)\sim 1 + h - f_3 < 1$.  Thus, in
  the perturbative regime, the speed of scalar perturbations is indeed
  less than unity. Again, one can see the importance of consistency
  conditions for the counter-terms. (Although such a scenario arises
  in the case of non-canonical scalar field inflationary models
  \cite{KInfPert}, classically it is not possible for canonical scalar
  field inflation.  Holonomy corrections, which have been used in
  \cite{HolonomyInfl} without ensuring consistency and
  anomaly-freedom, have been claimed to lead to a speed of sound much
  larger than unity and even divergent in some phases. This may
  indicate the inconsistency of the perturbation equations used
  there.)
\item[(ii)] From the corrected diffeomorphism constraint equation
(\ref{eq:dT0i}) one learns that in the absence of matter fields,
the metric perturbation decays more slowly compared to the classical
case. Assuming that $\bar{\alpha}$ is a slowly varying function, the
metric perturbations decay as
\begin{equation}
\Psi \propto \frac{1}{a^{(1 + f) (1 + h)}}
\end{equation}
According to (\ref{h}), $f+h=2\bar{\alpha}'p/\bar{\alpha}<0$ (and $fh$
is subdominant compared to $f+h$ in the regime considered
here). Hence, the decay should happen more slowly, implying that the
inverse triad corrections enhance metric perturbations.
\item[(iii)] In the long wavelength limit, the second term in 
Eq.~(\ref{eq:FinPhi}) can be neglected and the perturbations can be
treated to be independent of the wave-number $|{\bf k}|$. Assuming
that $\bar{\alpha}$ is a slowly varying function, the Bardeen
potential is given by
\begin{equation}
\Psi \propto \epsilon_{_{sr}} a^{\frac{3}{3 - 2 n_{\alpha}} - \frac{1}{2}}
\end{equation}
%
where the constant of proportionality is determined by the choice of
quantum state defined at the initial epoch of inflation
\cite{CosmoPert}. Comparing this expression with the
corresponding classical equation (\ref{Phi-zeta}) suggests that the
primordial perturbations have a different behavior compared to their
classical counterpart. For $0 < n_{\alpha} < 3/2$, the quantum
perturbations enhance the primordial perturbations compared to the
classical one. For $n_{\alpha} > 3/2$, the perturbations decay and
will lead to tiny primordial perturbations. This suggests that
quantifying the primordial perturbations can, in principle, constrain
the value of $n_{\alpha} > 3/2$. We will discuss the implications of
this in a future publication.
\item[(iv)] Note that in arriving at the above perturbation equation, we have
not used the perturbed Klein-Gordon equation (\ref{GI_KG}). To obtain
the primordial power spectrum, we would need to consider the combined
evolution of the scalar and the metric perturbations $\delta
\varphi^{{\rm GI}}$ and $\Psi$.
\end{enumerate}

\subsection{Isocurvature perturbations}

For canonical scalar fields in the classical theory, ${\cal R}$
is conserved on large scales implying that the perturbations are
adiabatic. It can be shown \cite{CosmoPert} that
\begin{equation}
\dot{\cal R}_{{\rm class}} = \frac{\H}{4 \pi G \dot{\bar{\varphi}}^2}  \Delta \Psi
~~~ \stackrel{k \to 0}\longrightarrow ~~~ 0 \, .
\end{equation}
Thus, on these scales ${\cal R}$ is conserved. However quantum effects
modify this and lead to small entropic perturbations. Taking a time
derivative of Eq.~(\ref{GI_CurvPert1}) and using the diffeomorphism
constraint (\ref{eq:dT0i}) leads to
\begin{eqnarray}
\dot{\cal R} &=& \dot{\Psi} + \f{1+ f}{1 + f_1}\f{\bar{\nu}}{\bar{\alpha}}
\left[ \ddot{\Psi} + \left(\H (1 + f) (1 + h) + \frac{\dot\H}{\H} - 2 \frac{\ddot{\bar{\varphi}}}{\dot{\bar{\varphi}}} \right) \dot{\Psi} \right. \\
& & \left. + 2 (1 + f) (1 + h) \left(\dot\H - \H  \frac{\ddot{\bar{\varphi}}}{\dot{\bar{\varphi}}} \right) \Psi
+ \H \left((1 + f) (1 + h)\right)^{\centerdot} \Psi \right] \frac{\H}{4 \pi G \dot{\bar\varphi}^2} \nonumber \\
& & + \left(\f{(1+ f)\bar{\nu} }{(1 + f_1) \bar{\alpha}}\right)^{\centerdot} \left(\dot{\Psi} + \H (1 + f) (1 + h) \Psi \right) \frac{\H}{4 \pi G \dot{\bar\varphi}^2} \nonumber
\end{eqnarray}
Eqs.~(\ref{eq:FinPhi}) and (\ref{KGBG}) in the long
wavelength limit imply
\begin{equation}
\dot{\cal R} \simeq \f{1+ f}{1 + f_1}\f{\bar{\nu}}{\bar{\alpha}}
\left[ 
 \left( \frac{3 \bar{\alpha}^{\prime} \bar{p}}{\bar{\alpha}} -
\frac{4 \bar{\nu}^{\prime} \bar{p}}{\bar{\nu}} - 3 \frac{f - f_3}{1 + f_3} \right) \dot\Psi
- \dot{h} (1 + f) \Psi
\right]
\frac{\H^2}{4 \pi G \dot{\bar\varphi}^2}
\end{equation}
{}from which we infer that ${\cal R}$ is not conserved on large
scales. Thus, perturbations generated during a single scalar field
epoch are no longer purely adiabatic --- perturbations contain a small
entropic contribution. Although such an effect has been seen earlier
in a class of Lorentz violating models \cite{GaugeInvTransPlanck}, our
example here is the first to show that the primordial perturbations
from inflation are not purely adiabatic and always contain a small
entropic perturbation.
(Note that anomaly-freedom of our equations \cite{ConstraintAlgebra}
ensures that there are no violations of Lorentz symmetry, although the
specific form of symmetries can be quantum corrected.)

\section{Discussion}

Effective constraints lead to quantum corrections in equations of
motion as well as in gauge transformations. Applied to general
relativity, it is not only the dynamics but also the underlying
space-time structure and the notion of covariance which are affected
by quantum corrections. In \cite{ConstraintAlgebra} it was shown that
the full constraint algebra, i.e.\ the set of structure functions, of
canonical quantum gravity changes when quantum corrections of a loop
quantization are included. While no gauge freedom is destroyed in the
anomaly-free quantization used, the algebra is not the classical
hypersurface deformation algebra as originally derived by Dirac for
classical general relativity. The quantum corrected equations used are
generally covariant, but the symmetry type of the underlying
covariance is quantum corrected.
A determination of the full gauge algebra of
quantum gravity would require going beyond the leading perturbative
order, which is not available so far. But the results of
\cite{ConstraintAlgebra} show that quantum corrections in the algebra
must arise.
In particular, this implies that terms in an effective
action of canonical quantum gravity cannot be simply of higher
curvature form. There must be additional effects such as non-local
terms or non-commutative manifold structures.

In \cite{ConstraintAlgebra} as well as in this paper only the inverse
triad type of quantum corrections due to the effects of
\cite{QSDI,QSDV} is considered. There are additional quantum
corrections, one due to the use of holonomies and a generic one due to
quantum back-reaction of fluctuations, correlations and higher moments
of a quantum state \cite{Karpacz}. For these corrections no
anomaly-free version for perturbative inhomogeneities has been found
yet, which indicates that there are severe consistency restrictions
especially for holonomy corrections. (Quantum back-reaction is
generic, such that the existence of consistent deformations is
guaranteed by the work on effective gravity e.g.\ in
\cite{EffectiveNewton,EffectiveGR}.)  Despite this incomplete status
of quantum gravity corrections, the different structures of the three
types of corrections shows that it is not possible to cancel
corrections from one type, such as the inverse triad corrections used
here, by corrections of the other types. Thus, corrections of the form
discussed here must be present in any cosmological perturbation theory
based on loop quantum gravity.

Given anomaly-freedom of the constraints, it is possible to construct
gauge invariant variables and recast the equations of motion in an
entirely gauge invariant manner. Here, this analysis was done for the
perturbative constraints of \cite{ConstraintAlgebra}, incorporating
inverse triad corrections of loop quantum gravity in a way which, to
leading orders, is anomaly-free. The final equations
(\ref{DiffEqFinal}), (\ref{HamEqFinal}), (\ref{GI_KG}),
(\ref{GI_Diag}), and (\ref{GI_OffDiag}):
\begin{eqnarray}
\D_c\left(\dot\Psi+\H(1+f)\Phi\right)&=&4\pi
G\f{\bar{\alpha}}{\bar{\nu}}\dot\varphi \D_c\delta\varphi^{\rm GI}\label{Discussion_Diff}\\
\Delta(\bar{\alpha}^2
\Phi)-3\H(1+f)\left[\dot\Psi+\H\Phi(1+f)\right]&=&4\pi
G\f{\bar{\alpha}}{\bar{\nu}}(1+f_3)\left[\dot{\bar{\varphi}}
\delta\dot\varphi^{\rm
GI}-\dot{\bar{\varphi}}^2(1+f_1)\Phi\nonumber\right.\label{Discussion_Ham}\\
&&+\left.\bar{\nu} \bar{p} V_{,\varphi}(\bar{\varphi})
\delta\varphi^{\rm GI}\right]\\
\ddot\Psi+\H\left[2\dot\Psi\left(1-\f{\bar{\alpha}^\prime
\bar{p}}{\bar{\alpha}}\right)+\dot\Phi(1+f)\right]
&+&\left[\dot\H+2\H^2\left(1+f^\prime \bar{p} -
\f{\bar{\alpha}^\prime \bar{p}}{\bar{\alpha}}\right)\right]\Phi(1+f)\nonumber\label{Discussion_Ray}\\
&=&4\pi
G\f{\bar{\alpha}}{\bar{\nu}}\left[\dot\varphi\delta\dot\varphi^{\rm
GI}-\bar{p}\bar{\nu} V_{,\varphi}(\bar{\varphi})\delta\varphi^{\rm GI}\right]\\
\D_a\D^i(\bar{\alpha}^2\left(\Phi-\Psi(1+h)\right))&=&0\label{Discussion_OffDiag}\\
\delta \ddot \varphi^{\rm GI}+2 \H \delta \dot \varphi^{\rm GI}
\left(1 - \f{\bar{\nu}^\prime \bar{p}}{\bar{\nu}}-g_1^\prime
\bar{p}\right)\!\!\! &-&\!\!\!  \bar{\alpha}^2 
(1-f_3)\Delta \delta \varphi^{\rm GI} +\bar{\nu} \bar{p}
V,_{\varphi\varphi}(\bar{\varphi})\delta \varphi^{\rm GI}
\label{Disscussion_KG}\\
+2\bar{\nu} \bar{p} V,_{\varphi}(\bar{\varphi})(1+f_1) \Phi
\!\!&-&\!\! \dot{\bar{\varphi}}\left[(1+f_1)\dot\Phi
+3(1+g_1)\dot\Psi\right]-2\H\dot{\bar{\varphi}} (f_3^\prime
\bar{p})\Phi=0\nonumber
\end{eqnarray}
collected here from the last two sections are manifestly gauge
invariant and reproduce the classical perturbed Einstein's
equation if one omits the quantum corrections.  Besides gauge
invariance, there is also a consistency issue which arises since,
on general grounds, there are three unknown scalar functions
subject to five equations. Moreover Eqs.~(\ref{Discussion_Diff})
and (\ref{Discussion_OffDiag}) can be used to eliminate two of
these functions in terms of just one, say $\Phi$, which should
satisfy the three remaining equations.

Another perspective on closure of the equations of motion is given
by considering that the Klein-Gordon equation
(\ref{Disscussion_KG}) is not independent. In the covariant
formalism, it results from the energy conservation equation for
the matter field: $\nabla_\mu T^{\mu\nu}=0$, the counterpart of
the Bianchi identity of the gravitational sector. The latter
equation is automatically satisfied by construction of the
Einstein tensor. For this reason, the Klein-Gordon equation can be
expressed in terms of the other equations and their derivatives.
In the canonical formulation such an argument, referring to the
Bianchi identity, is not available, especially at
the effective level, for it is a priori not clear what kind of
action might correspond to the quantum corrected constraints.
Nonetheless, one can use the form of the Bianchi identity as
guidance to explicitly check the redundancy of the Klein-Gordon
equation. In addition, the Hamiltonian constraint equation is
indeed a constraint, which restricts initial data, rather than a
dynamical equation as it does not contain second order time
derivatives. If the constraint is also preserved dynamically it
does not break the consistency of the equations of motion.

In the canonical setting, this is realized if the constraints are
first class, which is the case in the situation at hand to the orders
considered. In fact, closure of the constraint algebra guarantees both
gauge invariance of the equations of motion derived here and their
consistency.  In one of the equations, specifically in the
diffeomorphism constraint equation (\ref{DiffEq}), we have explicitly
shown that all gauge dependent terms mutually cancel. Similar
straightforward but tedious calculations for the other equations can
be performed, but for brevity we did not present them here.  Given the
closure of the constraint algebra, such an explicit demonstration of
vanishing of the gauge terms becomes unnecessary, although it may
still serve as an independent consistency check.

In the light of this, one can arrive at the final gauge invariant
equations of motion using the following shortcut. After computing the
time derivative for the corresponding conjugate momentum using the
Poisson bracket, it is possible to keep only the gauge invariant parts
of the variables, dropping all the gauge terms ($B$, $E$ and
$(\delta\varphi^{\rm GI}-\delta\varphi)$ in our case). Although
setting $B=0=E$ in the metric would amount to the longitudinal gauge,
this procedure is not equivalent to fixing the longitudinal gauge
prior to deriving the equations of motion. By doing so, one would
loose variational equations by off-diagonal metric components and thus
control on the off-diagonal spatial Einstein equation. Without that
equation, the relationship between $\Phi$ and $\Psi$ would remain
undetermined and one could only `borrow' the relation between $\Phi$
and $\Psi$ from the classical picture. The latter has proven to be
incorrect at the effective level, as can be seen from
Eq.~(\ref{GI_OffDiag}). Also effects of counter-terms, required for
anomaly-freedom, could not be seen in a gauge-fixed analysis, which in
general makes such effective equations inconsistent. The same remarks
apply to other gauge choices, such as uniform gauge.

Terms which arise from a complete treatment of all gauge properties,
but which could not be seen in a gauge-fixed analysis, do have
physical implications. As an example, we can derive an evolution
equation for curvature perturbations by subtracting
Eq.~(\ref{HamEqFinal}) divided by $1+f_3$ from Eq.~(\ref{GI_Diag}) in
such a way that matter perturbations are canceled for a stiff fluid
presented by the free scalar in the case $V(\varphi)=0$:
\begin{eqnarray}
&& \ddot{\Psi}-\frac{\Delta(\alpha^2\Phi)}{1+f_3}+ \H \left(\left(2\left(
1-\frac{\alpha'p}{\alpha}\right)+3\frac{1+f}{1+f_3}\right)\dot{\Psi}+
(1+f)\dot{\Phi}\right)\nonumber\\
&+&
\left(
\dot{\H}\frac{2-\frac{\nu'p}{3\nu}}{1-\frac{\nu'p}{3\nu}}+
\H^2\left(2\left(1+f'p-\frac{\alpha'p}{\alpha}\right)+
3\frac{(1+f)^2}{1+f_3}-
\frac{1+\frac{\alpha'p}{\alpha}}{1-
\frac{\nu'p}{3\nu}}\right)\right)\Phi =0\,. \label{FreeModes}
\end{eqnarray}
This equation is similar to (\ref{eq:FinPhi}), but more special
because it was derived assuming a stiff fluid.  For long-wave length
modes one can ignore the Laplacian of $\Phi$ and arrive at an ordinary
differential equation in time for only $\Psi$ if we also use the
relation $\Phi=\Psi(1+h)$. The classical equation would then be solved
by a decaying function $\Psi(\eta)$ as well as a constant mode because
the scale factor for the stiff fluid case satisfies the classical
equation $\dot{\H}+2\H^2=0$ which makes the coefficient of $\Phi$ in
the evolution equation vanish. This conservation of power on large
scales \cite{PowerLargeScales} can be demonstrated for any perfect
fluid by eliminating stress-energy components from the gauge-invariant
equations. It can also be shown to be a direct consequence of the
classical conservation law.

With quantum corrections, however, the coefficient of $\Phi$ does not
cancel exactly and the constant mode disappears. First, the background
equation is now corrected to
\begin{equation}
 \dot{\H}= -2\H^2\left(1-\frac{1}{2}
\frac{(\bar{\alpha}\bar{\nu})'\bar{p}}{\bar{\alpha}\bar{\nu}}\right)
\end{equation}
which follows from a combination of (\ref{FriBG}) and (\ref{RayBG}) in
the case of a vanishing potential. With this, and using our
perturbativity assumptions on the correction functions, the
coefficient of $\Phi$ in (\ref{FreeModes}) is
\begin{equation}
 \H^2\left(2f'\bar{p}+6f-3f_3-\frac{\bar{\alpha}'\bar{p}}{\bar{\alpha}}
 +\frac{5}{3}\frac{\bar{\nu}'\bar{p}}{\bar{\nu}}\right)
\end{equation}
which does not have to vanish even if the anomaly
cancellation conditions are used. This confirms the conclusions of
\cite{InhomEvolve} which initially were based on a gauge-fixed
treatment in longitudinal gauge.
However, here the signs of the correction terms are different (for
instance, $f'<0$ while $f>0$) such that it depends on the regime
whether power is enhanced or suppressed.

The gauge invariant equations of motion can now be used to describe
evolution of the curvature perturbations, e.g. during cosmological
inflation. The small quantum corrections accumulated during a
sufficiently long inflationary phase may potentially lead to
detectable imprints on the surface of last scattering and be observed
in the Cosmic Microwave Background.

\section*{Acknowledgements}

This work was supported in part by NSF grants PHY0748336 and PHY0456913.
MK was supported by NSF grant PHY0114375.
SS wishes to thank Roy Maartens and Kevin Vandersloot for discussions.
He is being supported by the Marie Curie Incoming International Grant
IIF-2006-039205. GH thanks IGC, Penn State where a part of this work
was completed. GH is partially supported by NSERC of Canada.

\section*{Appendix}

\begin{appendix}

In this appendix we present the determination of gauge invariant
equations of motion, done in the main text in the presence of quantum
corrections, for the classical canonical theory. Resulting equations
can be seen to agree with \cite{CosmoPert}.

    \section{Gauge invariant equations of motion}
\label{Sec_EOM}

As we have seen earlier, the background phase space variables as well
as the background lapse and shift remain unchanged under the
infinitesimal gauge transformations (up to the second order).
Therefore the equations of motion governing them are also gauge
invariant. In the perturbed context, both the canonical fields and the
Lagrange multipliers gauge-transform in a non-trivial way, and
equations of motion should be formulated for the gauge-invariant
variables.  Here, we systematically derive the background and
perturbed canonical equations of motion of classical linearized
gravity. We also discuss the equivalence of the canonical system of
equations and Einstein's equation.

        \subsection{Background equations}

Since the background shift vector is zero, the background
diffeomorphism constraint is identically satisfied. Therefore
background equations are generated only by the background
Hamiltonian constraint
\begin{equation}\label{BG_Ham}
H^{(0)}[\bar N]=V_0\bar N\left[-\f{3\sqrt{\bar{p}} \bar{k}^2}{8\pi
G}+\f{\bar\pi^2}{2 \bar{p}^{3/2}}+\bar{p}^{3/2} V(\bar\varphi)\right],
\end{equation}
where $\bar{k}$ and $\bar{p}$ are the background (diagonal) components
of the extrinsic curvature and triad respectively, $\bar\varphi$ and
$\bar\pi$ are the background matter field and its conjugate
momentum. (The parameter $V_0=\int {\rm d}^3x$ is again the coordinate
volume of space or of a compact region in which the perturbations are
introduced.)  The background lapse, $\bar N=a\equiv\sqrt{\bar{p}}$,
corresponds to the conformal time.

The Hamiltonian (\ref{BG_Ham}) gives rise to the constraint
equation
\begin{equation}\label{BG_ConstrEq}
0=\f{\D H^{(0)}[\bar N]}{\D \bar N}=V_0\left[-\f{3\sqrt{\bar{p}}
\bar{k}^2}{8\pi G}+\f{\bar{\pi}^2}{2 \bar{p}^{3/2}}+
\bar{p}^{3/2} V(\bar{\varphi})\right]
\end{equation}
and two pairs of Hamilton's equations of motion
\begin{eqnarray}
\dot{\bar k}&=&\{\bar{k},H^{(0)}[\bar N]\}=-\f{\bar N \bar{k}^2}{2\sqrt{\bar{p}}}
+4\pi G\bar N\left(-\f{\bar{\pi}^2}{\bar{p}^{5/2}}+
\sqrt{\bar{p}}V(\bar{\varphi})\right) \label{BG_kDot}\\
\dot{\bar p}&=&\{\bar{p},H^{(0)}[\bar N]\}=2\bar N \sqrt{\bar{p}}\bar{k}
\label{BG_pDot}
\end{eqnarray}
for the gravitational variables and
\begin{eqnarray}
\dot{\bar{\varphi}}&=&\{\bar{\varphi},H^{(0)}[\bar N]\}=
\f{\bar N \bar{\pi}}{\bar{p}^{3/2}}\label{BG_phiDot}\\
\dot{\bar{\pi}}&=&\{\bar{\pi},H^{(0)}[\bar N]\}=-\bar N \bar{p}^{3/2}
V_{,\varphi}(\bar\varphi) \label{BG_piDot}
\end{eqnarray}
for the matter field. Note that the lapse $\bar N=\sqrt{\bar{p}}$ should
be fixed consistently after computing the Poisson brackets. Then,
as follows from Eq.~(\ref{BG_pDot}), the background extrinsic
curvature is nothing but the conformal Hubble parameter
\be\label{Hubble}%
\bar{k}=\f{\dot{\bar{p}}}{2\bar{p}}\equiv\f{\dot a}{a}=:\H,
\ee
whereas the field momentum is given by
\be\label{pi_phiDot}
\bar\pi=\bar{p}\dot{\bar{\varphi}}\,.
\ee
Using the two relations above in (\ref{BG_ConstrEq}) yields the
Friedmann equation
\begin{equation}
\H^2=\f{8\pi G}{3}\left(\f{\dot{\bar{\varphi}}^2}{2}+
\bar{p}V(\bar{\varphi})\right)\label{FriBG_Cl},
\end{equation}
whose right hand side is proportional to the background matter
energy density. The gravitational equation (\ref{BG_kDot}) should
be recognized as the Raychaudhuri equation
\begin{equation}
\dot\H=\H^2-4\pi G\dot{\bar{\varphi}}^2\label{RayBG_Cl},
\end{equation}
written in conformal time. Finally, combining the matter
Hamilton's equations into a single second order equation results
in the Klein-Gordon equation
\begin{equation}
\ddot{\bar{\varphi}}+2\H\dot{\bar{\varphi}}+
\bar{p}V_{,\varphi}(\bar{\varphi})=0\label{KGBG_Cl}.
\end{equation}
The unusual factor of 2 (instead of the standard one, 3) in the
second term is again due to the choice of conformal time.

        \subsection{Perturbed equations}
Perturbed equations of motion are generated by (the second order
part of) both Hamiltonian and diffeomorphism constraints. Using
the background and perturbed lapse function, the gravitational
Hamiltonian constraint can be written as
\begin{equation}
\label{ClassPertHamConst} H^{(2)}_{\rm grav}[N] = \frac{1}{16\pi
G}\int \mathrm{d}^3x \left[\bar{N} {\mathcal H}^{(2)} +
 \delta N{\mathcal H}^{(1)}\right] \,,
\end{equation}
with the first- and second-order parts of the Hamiltonian density
given by (\ref{PertHamConst_Grav})
and
\begin{eqnarray}
{\mathcal H}^{(2)} &=& \sqrt{\bar{p}} \delta K_c^j\delta
K_d^k\delta^c_k\delta^d_j - \sqrt{\bar{p}} (\delta
K_c^j\delta^c_j)^2 -\frac{2\bar{k}}{\sqrt{\bar{p}}} \delta
E^c_j\delta K_c^j
\nonumber\\
&& \quad -\frac{\bar{k}^2}{2\bar{p}^{3/2}} \delta E^c_j\delta
E^d_k\delta_c^k\delta_d^j +\frac{\bar{k}^2}{4\bar{p}^{3/2}}(\delta
E^c_j\delta_c^j)^2 -\frac{\delta^{jk}
}{2\bar{p}^{3/2}}(\partial_c\delta E^c_j) (\partial_d\delta
E^d_k)\,,
\end{eqnarray}
respectively. The second order part of the diffeomorphism
constraint reads
\begin{equation} \label{PertDiffConst_Grav}
D^{(2)}_{\rm grav}[N^a] = \frac{1}{8\pi
G}\int_{\Sigma}\mathrm{d}^3x\delta N^c
\left[\bar{p}\partial_c(\delta^d_k \delta K^k_d)
-\bar{p}(\partial_k\delta K^k_c)- \bar{k} \delta_c^k(
\partial_d \delta E^d_k)\right].
\end{equation}
In the matter sector, we similarly have the Hamiltonian constraint
\begin{eqnarray} \label{PertHamConst_Matt}
H^{(2)}_{\rm matter}[{N}] = \int_\Sigma\mathrm{d}^3x \left[\bar{N}
\left({\mathcal H}_\pi^{(2)}+{ \mathcal H}_\nabla^{(2)}+ {\mathcal
H}_\varphi^{(2)}\right)+\delta N \left( {\mathcal
H}_\pi^{(1)}+{\mathcal H}_\varphi^{(1)} \right) \right]
\end{eqnarray}
with the densities (\ref{SFHamConstH1})
and
\begin{eqnarray} \label{SFHamConstH2}
{\mathcal H}_\pi^{(2)} &=&
\frac{1}{2}\frac{{\delta{\pi}}^2}{\bar{p}^{3/2}} -\frac{\bar{\pi}
\delta{\pi}}{\bar{p}^{3/2}} \frac{\delta_c^j \delta
E^c_j}{2\bar{p}} +\frac{1}{2}\frac{\bar{\pi}^2}{\bar{p}^{3/2}}
\left( \frac{(\delta_c^j \delta E^c_j)^2}{8\bar{p}^2}
+\frac{\delta_c^k\delta_d^j\delta E^c_j\delta E^d_k}{4\bar{p}^2}
\right) ~, \nonumber\\
{\mathcal H}_\nabla^{(2)} &=& \frac{1}{2}\sqrt{\bar{p}}\delta^{ab}
\partial_a\delta \varphi \partial_b\delta \varphi \nonumber~,\\
{\mathcal H}_\varphi^{(2)} &=& \frac{1}{2}\bar{p}^{3/2}
V_{,\varphi\varphi}(\bar{\varphi}) {\delta\varphi}^2
+\bar{p}^{3/2} V_{,\varphi}(\bar{\varphi}) \delta\varphi
\frac{\delta_c^j \delta E^c_j}{2\bar{p}}
\nonumber\\
&&\quad + \bar{p}^{3/2} V(\bar{\varphi})\left( \frac{(\delta_c^j
\delta E^c_j)^2}{8\bar{p}^2} -\frac{\delta_c^k\delta_d^j\delta
E^c_j\delta E^d_k}{4\bar{p}^2} \right) ~,
\end{eqnarray}
along with the diffeomorphism constraint
\begin{equation} \label{SFPertDiffConst}
D^{(2)}_{\rm matter}[N^a] = \int_{\Sigma}\mathrm{d}^3x\delta N^c
\bar{\pi}\partial_c\delta \varphi.
\end{equation}
In the expressions above, the triad perturbation has the form
(\ref{Triad}), whereas the perturbations of lapse and shift are
given by (\ref{PertLapseShift}).

Below we formulate the perturbed equations of motion purely in terms
of the gauge invariant variables. The equations, as before, are of two
types:\\ (i) Constraint equations, i.e. the Hamiltonian and
Diffeomorphism constraints and\\ (ii) Dynamical (Hamilton's)
equations. \\ The latter are those for the matter variables (one
second order equation) and for the gravitational variables (two
independent second order equations: diagonal and off-diagonal).  The
following auxiliary relation will be useful for deriving
gauge-invariant equations:
\begin{eqnarray}\label{GaugeAux_Cl}
\delta E^a_i=&-&2\bar{p}\Psi\delta_i^a \nonumber\\
&-&2\H \bar{p}(B-\dot E) \delta_i^a +\bar{p}(\delta_i^a\Delta-\D^a\D_i)E \\
 \delta K_a^i=&-&\delta_a^i\left[\dot\Psi+\H\left(\Psi+\Phi\right)\right]\nonumber\\
&-&\delta_a^i\dot\H(B-\dot E)+\D_a \D^i\left[\H E - (B-\dot
E)\right],\nonumber
\end{eqnarray}
where the first line of each equation contains only gauge
invariant terms.

\subsubsection{Diffeomorphism constraint equation} Varying the smeared
diffeomorphism constraint with respect to the shift perturbation
yields the diffeomorphism constraint equation (the space-time
Einstein equation):
\begin{equation}
0=8\pi G\f{\delta D[\delta N^c]}{\delta (\delta
N^c)}=\bar{p}\left(\D_c(\delta K_a^i\delta_i^a)-\D_k(\delta
K_c^k)\right)-\bar{k}\D_d\delta E^d_k \delta_c^k+8\pi G \bar{\pi}
\D_c\delta\varphi.
\end{equation}
Using the gauge invariant variables defined in
(\ref{Bardeen}) and (\ref{GI_varphi_Cl}), this
equation can be rewritten as
\begin{equation}\label{DiffEqClass}
\D_c\left[\dot\Psi+\H\Phi-4\pi G\dot{\bar{\varphi}}\delta\varphi^{\rm
GI}\right]+({\rm gauge \, \,terms})=0,
\end{equation}
where the `gauge terms' are
\[
2(B-\dot E)\D_c\left[-\dot\H+\H^2-4\pi G \dot{\bar{\varphi}}^2\right]
\]
The expression inside the square brackets is nothing but the
background Raychaudhuri equation (\ref{RayBG_Cl}). Thus all gauge
dependent terms vanish and the diffeomorphism constraint equation
takes the form
\begin{equation}\label{DiffEqFinalClass}
\D_c\left[\dot\Psi+\H\Phi\right]=4\pi G\dot{\bar{\varphi}}
\D_c\delta\varphi^{\rm GI},
\end{equation}
whose right hand side should be recognized as
the gauge invariant space-time component of the perturbed matter
stress-energy tensor $-4\pi G a^2 \delta T_S^T$.

\subsubsection{Hamiltonian constraint equation}


The Hamiltonian constraint equation is obtained by variation with
respect to the lapse perturbation:
\begin{eqnarray}
\f{\delta H}{\delta (\delta N)}&=&\f{1}{16\pi G}\left[-4
\bar{k}\sqrt{\bar{p}}(\delta K_a^i \delta_i^a)-
\f{\bar{k}^2}{\sqrt{\bar{p}}}\delta
E_i^a\delta_a^i+\f{2}{\sqrt{\bar{p}}}\D_a\D^i\delta
E^i_a\right]\nonumber\\
&&+\f{\bar\pi\delta\pi}{\bar{p}^{3/2}}-\left(\f{\bar\pi^2}{2\bar{p}^{3/2}}
-\bar{p}^{3/2}V(\bar\varphi)\right)\f{\delta E_i^a
\delta^i_a}{2\bar{p}}+\bar{p}^{3/2}V_{,\varphi}(\bar\varphi)
\delta\varphi\nonumber\\
&=&0\,.
\end{eqnarray}
Dividing both sides by $\sqrt{\bar{p}}$ allows one to replace the
background extrinsic curvature with the Hubble rate. Then
eliminating the field momentum and its perturbation in terms of
the time derivatives of the scalar field (see
(\ref{Pert_EoM_varphi})) and using the auxiliary expressions
(\ref{GaugeAux_Cl}), one arrives at the gauge invariant Hamilton
constraint equation (perturbed Friedmann equation)
\begin{eqnarray}\label{HamEqFinalClass}
\Delta\Phi-3\H\left[\dot\Psi+\H\Phi\right]&=&4\pi
G\left[\dot{\bar\varphi}\delta\dot\varphi^{\rm GI}-\dot{\bar\varphi}^2\Phi+
\bar{p}
V_{,\varphi}(\bar\varphi)\delta\varphi^{\rm GI}\right]
\end{eqnarray}
Again, the right hand side is nothing but the time-time
component of the perturbed stress-energy tensor.

\subsubsection{Hamilton's equations}

The perturbed dynamical matter equations are computed using the
Poisson bracket, giving rise to (\ref{Pert_EoM_varphi}) and
(\ref{Pert_EoM_pi}).
Expressing $\delta\pi$ from the first equation and substituting
into the second one, with the help of the background equations,
the Klein-Gordon equation can be cast in the gauge
invariant form
\bq\label{GI_KG_Cl}%
\delta \ddot \varphi^{\rm GI}+ 2 \H \delta \dot \varphi^{\rm GI} -
\Delta \delta \varphi^{\rm GI} + \bar{p}
V,_{\varphi\varphi}(\bar{\varphi})\delta \varphi^{\rm GI}+2 \bar{p}
V,_{\varphi}(\bar{\varphi}) \Phi
-\dot{\bar{\varphi}}\left(\dot\Phi+3\dot\Psi\right)=0\,.
\eq%
Similarly one can arrive at the spatial Einstein equations. The
first order equations are given by
\begin{eqnarray}
 \delta \dot K_a^i &\equiv& \{\delta K_a^i,
H^{(2)}[N]+D^{(2)}[N^a]\}\nonumber\\&=&\f{\bar
N}{\bar{p}^{3/2}}\left[-\bar{k}\bar{p}\delta K_a^i-
\f{\bar{k}^2}{2}\delta E_k^d
\delta^i_d\delta_a^k+\f{\bar{k}^2}{4}\delta E_k^d\delta_d^k\delta_a^i+
\f{\delta^{ik}}{2}\D_a\D_d\delta E_k^d
\right]\label{Pert_EoM_delK} \\
&+&\f{\delta N
\bar{k}^2}{2\sqrt{\bar{p}}}\delta_a^i-
\f{1}{\sqrt{\bar{p}}}\D_a\D^i\delta N+8\pi G
\f{\delta H^{(2)}_m[N]}{\delta(\delta E_i^a)}\nonumber\\
\delta \dot E_i^a &\equiv& \{\delta E_i^a,
H^{(2)}[N]+D^{(2)}[N^a]\}\label{Pert_EoM_delE}\\
&=&\f{\bar N}{\sqrt{\bar{p}}}\left[\bar{p}\delta
K_c^j\delta_i^c\delta_j^a-\bar{p}(\delta
K_c^j\delta_j^c)\delta_i^a-\delta E_i^a\right] +2\delta N
\bar{k}\sqrt{\bar{p}}\delta_i^a+\bar{p}\left(\delta_i^a\D_c\delta N^c-
\D_i\delta N^a\right)\nonumber
\end{eqnarray}
The combined second order  equation naturally decouples into two
independent equations: diagonal and off-diagonal. After a tedious
but rather straightforward computation, taking into account the
background equations of motion, the former equation takes the
form
\begin{eqnarray}\label{GI_Diag_Cl}
\ddot\Psi+\H\left(2\dot\Psi+\dot\Phi\right)
+\left(\dot\H+2\H^2\right)\Phi=4\pi
G\left(\dot{\bar{\varphi}}\delta\dot\varphi^{\rm GI}-\bar{p}
V_{,\varphi}(\varphi)\delta\varphi^{\rm GI}\right)\,.
\end{eqnarray}
In the absence of anisotropic stress in the matter sector, which
is the case for the scalar field, the gauge invariant part of the
off-diagonal equation reads:
\begin{equation}\label{GI_OffDiag_Cl}
\D_a\D^i\left[\Phi-\Psi\right]=0,
\end{equation}
which implies $\Phi=\Psi$ for the Bardeen potentials.


\section{Anomaly Cancellation Conditions}
\label{App:Anomaly}
In this appendix, we summarize the set of anomaly cancellation
conditions containing counter-term coefficients. These conditions, in
turn, determine the coefficients in terms of primary quantum
correction functions $\alpha$, $\nu$ and $\sigma$, and they impose
restrictions on the primary corrections. We note from the expressions
of quantum corrected Hamiltonian densities (\ref{HamGravQDens}),
(\ref{HamMatterQDens}) that there are three such functions ($f$, $g$,
$h$) in the gravitational sector, six ($f_1$, $f_2$, $g_1$, $g_2$,
$g_3$, $g_5$) in the kinetic sector and two ($f_3$, $g_6$) in the
potential sector of scalar matter. Thus for the system under
consideration we have a total of eleven initially undetermined
functions contained in the counter-terms. These free functions are
constrained by anomaly cancellation.

Invariance of counter-terms under diffeomorphisms
\cite{ConstraintAlgebra} led to four conditions
\begin{equation} \label{FinalConditionsDiffeo}
g = -2f  ~,~ f_2 = 2 f_1 ~,~ g_2=g_1 ~,~ g_3=2g_2 = 2g_1
\end{equation}
among these coefficients.
These equations trivially lead to the solutions for $g$, $f_2$, $g_2$
and $g_3$, leaving seven functions to be determined.
Cancellation of anomaly terms from the Poisson bracket between
Hamiltonian constraints led to two independent conditions
 ($\mathcal{G}_1^{\varphi}=0$, $\mathcal{G}_2^{\varphi}=0$ of
\cite{ConstraintAlgebra}) from the gravitational sector. These two conditions
along with the equation (\ref{FinalConditionsDiffeo}) imply
\begin{eqnarray} \label{FinalAnoCanGravEqn1}
h &=& -f+ 2 \frac{\bar{\alpha}'\bar{p}}{\bar{\alpha}} ~,\\
\label{FinalAnoCanGravEqn2}
2 \bar{p}f' &=& -\frac{\bar{\alpha}'\bar{p}}{\bar{\alpha}}~.
\end{eqnarray}
These two equations explicitly solve $f$ and $h$ in terms of the
primary correction function $\alpha$. In particular, for the given
form of $\bar\alpha =1+c_{\alpha}\left(\ell_{\rm
P}^2/\bar{p}\right)^{n_{\alpha}}+\cdots$,
(\ref{FinalAnoCanGravEqn2}) has the solution
\begin{equation}\label{FSol}
f = \frac{1}{2 n_{\alpha}} \frac{\bar{\alpha}'\bar{p}}{\bar{\alpha}}
\end{equation}
Thus, there are only five remaining functions that need to be
determined. Three of the matter anomaly cancellation conditions
($\mathcal{D}_1=0$, $\mathcal{D}_2=0$, $\mathcal{B}_4^{\varphi}=0$ of
\cite{ConstraintAlgebra}) lead to
\begin{eqnarray} \label{FinalAnoCanSFEqn1}
g_5 &=& - f_1 ~,~  g_6 = - g_1 ~,~ f_3 = f_1 -  g_1
\end{eqnarray}
whereas another two conditions
($\mathcal{B}_1^{\varphi}=0$, $\mathcal{B}_2^{\varphi}=0$
of \cite{ConstraintAlgebra}) solve $f_1$ and $g_1$ explicitly as
\begin{eqnarray} \label{FinalAnoCanSFEqn2}
f_1 &=& f -  \frac{\bar{\nu}'\bar{p}}{3\bar{\nu}} ~,\\
g_1 &=& \frac{\bar{\alpha}'\bar{p}}{3\bar{\alpha}}
-  \frac{\bar{\nu}'\bar{p}}{\bar{\nu}}
+  \frac{2}{9}\left(\frac{\bar{\nu}'\bar{p}}{\bar{\nu}}\right)^{'}p ~.
\end{eqnarray}
At this stage all counter-terms coefficients have been
determined. There is one remaining anomaly cancellation condition from the
matter sector, $\mathcal{D}_4 = 0$ of \cite{ConstraintAlgebra}, which
implies
\begin{equation} \label{f3diff}
 2\bar{p}f_3'+3f_3-3f=0
\end{equation}
 and thus requires the primary correction functions to satisfy
\begin{equation}\label{alphanuRelation}
\f{\bar\alpha^\prime\bar p}{\bar\alpha} +\f{\bar
p}{3}\left(\f{\bar\alpha^\prime\bar p}{\bar
\alpha}\right)^\prime-\f{\bar\nu^\prime\bar p}{\bar \nu} -\f{\bar
p}{9}\left(\f{\bar\nu^\prime\bar p}{\bar
\nu}\right)^\prime+\f{2\bar p^2}{9}\left(\f{\bar\nu^\prime\bar
p}{\bar \nu}\right)^{\prime\prime}=0.
\end{equation}
Independent of the counter-terms and the the requirement
(\ref{alphanuRelation}), we also have the relation
$\bar\alpha^2=\bar\nu\bar\sigma$ to be satisfied by the primary
correction functions. Thus, anomaly freedom of the constraint algebra
severely restricts the allowed form of primary quantum corrections
functions, but it does permit non-trivial forms of quantum
corrections.
\end{appendix}


\end{document}